\documentclass[format=acmsmall, review=false]{acmart}
\usepackage{acm-ec-20}
\usepackage{booktabs} % For formal tables
\usepackage[ruled]{algorithm2e} % For algorithms

\SetAlFnt{\small}
\SetAlCapFnt{\small}
\SetAlCapNameFnt{\small}
\SetAlCapHSkip{0pt}
\IncMargin{-\parindent}
\usepackage{cleveref}

\setcitestyle{acmnumeric}

\usepackage{amsmath}
\usepackage{amsfonts}
\usepackage{amsthm}
\usepackage{bbm}
\usepackage{indentfirst}
\usepackage{fancyhdr}
\usepackage{graphicx}
\usepackage{float}
\usepackage{tabu}
\usepackage{tikz}
\usetikzlibrary{patterns}

\usepackage{caption}
\usepackage{subcaption}

\newcommand{\p}[1]{\left({#1}\right)}
\newcommand{\ps}[1]{\left[{#1}\right]}
\newcommand{\pb}[1]{\left\{{#1}\right\}}
\newcommand{\set}[1]{\left\{#1\right\}}

\newcommand{\beq}[1]{\begin{equation}\label{#1}}
\newcommand{\eeq}{\end{equation}}
\usepackage{xcolor}

\newcommand{\bledit}[1]{#1}

\newcommand{\SCedit}[1]{#1}

\newcommand{\opt}{\ensuremath{\textrm{OPT}}}
\newcommand{\under}{\ensuremath{\textsc{Under}}}
\newcommand{\single}{\ensuremath{\textsc{Single}}}
\renewcommand{\over}{\ensuremath{\textsc{Over}}}
\newcommand{\surplus}{\ensuremath{\textsc{Surplus}}}
\newcommand{\nonfav}{\ensuremath{\textsc{NonFavorite}}}
\newcommand{\vw}{\ensuremath{\textsc{VW}}}
\newcommand{\core}{\ensuremath{\textsc{Core}}}
\newcommand{\tail}{\ensuremath{\textsc{Tail}}}

\newcommand{\efrev}{\ensuremath{\mathsf{EF}\text{-}\mathsf{Rev}}}
\newcommand{\bb}{\textbf{b}}

\newcommand{\EA}{\ensuremath{\textsc{EA}}}
\newcommand{\REA}{\mathsf{rand}-\EA}
\newcommand{\GEA}{\mathsf{ghost}-\EA}
\newcommand{\SP}{\ensuremath{\textsc{SP}}}
\newcommand{\ESP}{\ensuremath{\textsc{ESP}}}
\newcommand{\EFP}{\ensuremath{\textsc{EFP}}}
\newcommand{\RFP}{\mathsf{rand}-\EFP}
\newcommand{\RAP}{\mathsf{rand}-\EAP}
\newcommand{\GFP}{\mathsf{ghost}-\EFP}
\newcommand{\EAP}{\ensuremath{\textsc{EAP}}}
\newcommand{\GAP}{\mathsf{ghost}-\EAP}
\newcommand{\PP}{\ensuremath{\textsc{PP}}}
\newcommand{\spr}{\ensuremath{\textsc{SSP}}}
\newcommand{\FP}{\ensuremath{\textsc{FP}}}
\newcommand{\SFP}{\ensuremath{\textsc{SFP}}}

\newtheorem{theorem}{Theorem}

\newtheorem{mydef}[theorem]{Definition}

\newtheorem{cor}[theorem]{Corollary}
\newtheorem{lemma}[theorem]{Lemma}
\newtheorem*{blankthm}{\Cref{thm:main}}
\newtheorem*{blanklemma}{\Cref{lem:contCai}}
\newtheorem*{blanklemma2}{\Cref{lem:rand-ea-eq}}
\newtheorem*{blankdef}{\Cref{def:MPP}}
\newtheorem*{blankcor}{\Cref{cor:first-price}}

\numberwithin{theorem}{section}

\newcommand{\cB}{\mathcal{B}}
\newcommand{\cC}{\mathcal{C}}

\newcommand{\cF}{\mathcal{F}}

\newcommand{\cR}{\mathcal{R}}

\newcommand{\cU}{\mathcal{U}}
\newcommand{\cZ}{\mathcal{Z}}
\newcommand{\Rev}{\mathsf{Rev}}
\newcommand{\Wel}{\mathsf{Wel}}
\newcommand{\Util}{\mathsf{Util}}

\newcommand{\EE}{\mathbb{E}}
\newcommand{\RR}{\mathbb{R}}

\title[Multi-item Non-truthful Auctions Achieve Good Revenue]{Multi-item Non-truthful Auctions Achieve Good Revenue}

\author{Constantinos Daskalakis}
\affiliation{\institution{Massachusetts Institute of Technology}}
\email{costis@csail.mit.edu}
\author{Maxwell Fishelson}
\affiliation{\institution{Massachusetts Institute of Technology}}
\email{maxfish@mit.edu}
\author{Brendan Lucier}
\affiliation{Microsoft Research}
\email{brlucier@microsoft.com}
\author{Vasilis Syrgkanis}
\affiliation{Microsoft Research}
\author{Santhoshini Velusamy}
\affiliation{Harvard University}
\email{svelusamy@harvard.edu}

\begin{abstract}

 \hspace{8pt} We present a general framework for designing approximately revenue-optimal mechanisms for multi-item additive auctions, which applies to both truthful and non-truthful auctions. 
 Given a (not necessarily truthful) single-item auction format $A$ satisfying certain technical conditions, we run simultaneous item auctions augmented with a personalized entry fee for each bidder that must be paid before the auction can be accessed.
 These entry fees depend only on the prior distribution of bidder types, and in particular are independent of realized bids.
 We bound the revenue of the resulting two-part tariff mechanism using a novel geometric technique that enables revenue guarantees for many common non-truthful auctions that previously had none. Our approach adapts and extends the duality framework of Cai et al \cite{CDW16} beyond truthful auctions.
 
 Our framework can be used with many common auction formats, such as simultaneous first-price, simultaneous second-price, and simultaneous all-pay auctions.  
 Our results for first price and all-pay are the first revenue guarantees of non-truthful mechanisms in multi-dimensional environments, addressing an open question in the literature \cite{roughgarden2017price}.
 If all-pay auctions are used, we prove that the resulting mechanism is also credible in the sense that the auctioneer cannot benefit by deviating from the stated mechanism after observing agent bids.  This is the first static credible mechanism for multi-item additive auctions that achieves a constant factor of the optimal revenue.  If second-price auctions are used, we obtain a truthful $O(1)$-approximate mechanism with fixed entry fees that are amenable to tuning via online learning techniques. 

\end{abstract}

\begin{document}

\maketitle
\pagebreak
\tableofcontents

% \begin{itemize}
%     \item Intro orange paragraph
%     \item Discuss single item FP literature
%     \item Discuss multi item literature
%     \item Until this work, not even log n was known
% \end{itemize}
\pagebreak
\section{Introduction}
The first-price auction is the most common auction format in practice.  In the display ads market (the most significant usage of auctions in modern commerce), every major exchange utilizes first price mechanisms.  While the second-price auction guarantees optimal social welfare, the winner-pays-bid dynamic of first-price offers transparency in payment and, above all else, simplicity.  Until this paper, though, there was no guarantee that this auction achieved revenue comparable to the optimal when selling multiple items.

Multi-item auction revenue has been the subject of intense focus in the recent algorithmic mechanism design literature.  Much of this work explores an inherent tradeoff between mechanism optimality and simplicity.  Indeed, even when valuations are additive and independent across items, revenue-optimal mechanisms are known to be highly complex: they can require lotteries~\cite{Than04,Pavlov11,DDT13}, can exhibit non-monotone revenue~\cite{HP15}, and can be computationally difficult to compute~\cite{DDT14}.  This has motivated a relaxation of revenue-optimality, leading to a search for simple and robust auction formats that approximate the Bayesian optimal revenue; see e.g.~\cite{chawla2007algorithmic,chawla2010multi,Chawla2015,BILW14,yao2014n,chawla2016mechanism,CDW16,cai2017simple,cai2017learning,CaiZ19}.  This search has culminated in a line of work establishing that approximately optimal revenue can be obtained through the better of separate item auctions and two-part tariff auctions~\cite{yao2014n,chawla2016mechanism,cai2017simple}, wherein bidders are asked to pay an entry fee for the chance to bid on individual items.  Thus, by randomly selecting and running one of these two auctions, approximately optimal revenue can be achieved with a relatively simple mechanism.
%and then all those who pay participate in simultaneous single-item auctions. 
%Entry fees have been a well-studied topic in the auction literature~\cite{Milgrom1982,Mcafee1987,Meyer1993,Engelbrecht1993,Levin1994,Armstrong1999,Cao2010,chen2018auctions}, and this new line of work demonstrates that they can be very useful in achieving approximately optimal revenue in multi-dimensional settings. Moreover, the aforementioned auctions are truthful and can be implemented in a computationally efficient manner.

In the way of simplicity though, this line of work leaves several major stones unturned. Firstly, none of these results have been able to establish revenue guarantees for winner-pays-bid auctions, or any simple non-truthful auctions for that matter.  \bledit{For example, the two-part tariff auctions described above assume that items are sold via second-price auctions after the entry fee(s) have been paid.  On the other hand, a modified version of the first price auction that adds a reserve price (or minimum bid) has been shown to achieve approximately optimal revenue in the single-item setting \cite{HHT14}.  This enables us to say that the better of separate reserve-first-price auctions and two-part tariff auctions (with second-price auctions post-entry-fee) achieves good revenue.  However, generalizing the two-part tariff side of these results to winner-pays-bid mechanisms
%The first price auction with reserve Revenue guarantees for these mechanisms have been shown for the single-item setting \cite{HHT14}.  This result enables us to say that the better of 
%However, extending these results to the multi-item setting
%has historically met resistance due to the 
is not straightforward, in large part because of the
complexity of the bidding equilibria in 
%these 
winner-pays-bid
auctions.}  Additionally, %none of these studied mechanisms truly achieve the goal of the two-part tariff auction: to take a standard single-item auction format and simply add an entry fee. 
in all prior works, these two-part tariff mechanisms are either dynamic, requiring multiple rounds of communication with the agents \cite{chawla2016mechanism,cai2017simple,cai2017learning},
\footnote{In particular, these works provide guarantees for {\em sequential posted price mechanisms with entry fees}, wherein bidders interact with the mechanism in sequence. Each bidder is shown an entry fee 
%results about two-part tariff auctions.  That is, all previous research has been on mechanisms have involved greater complexity than simply adding an entry fee to a standard auction format used in practice
%(which depends on the set of remaining items) 
and individual item prices 
%(which might depend on the identity of the buyer) 
and is asked to pay the entry fee for the chance of purchasing any of the available items at its posted price.  Of course, any such mechanism has a static direct-revelation implementation (i.e., that solicits bids and simulates the sequential mechanism), but such an implementation sacrifices credibility.}
or involve entry fees that are not posted in advance. Rather, in these mechanisms, the fee presented to each agent is a function of the submitted bids of all other agents \cite{yao2014n,CDW16}. In the former case, the multi-round nature of the mechanisms can present implementation difficulties that static mechanisms bypass; as noted in \cite{AL18}, static mechanisms can be conducted rapidly and asynchronously, which yields several implementation benefits, supported by empirical evidence \cite{Athey2009}. In the latter case, since the entry fees are opaque functions of other player reports, the connection with the colloquial notion of ``entry fee'' is arguable.

We study a framework that addresses all of these concerns, establishing revenue guarantees for a large class of static mechanisms with \emph{fixed} entry fees, posted in advance to all agents.  Our framework takes the form of a general construction: given an arbitrary (not necessarily truthful) single-item auction format, our mechanism proceeds by first posting a fixed entry fee to each agent, independent of the realized bids and determined only based on the prior distributions from which values are to be drawn.  The mechanism then sells each item separately using the provided single-item auction format, where only agents who paid their respective entry fees are allowed to participate.  Any agents who did not pay the entry fee are instead simulated by the mechanism, and any items won by such simulated agents are discarded.  These ``ghost bidders'' ensure that the equilibria of the simultaneous single-item auctions are unaffected by the realization of which agents pay their respective entry fees.  We show that if the provided auction satisfies a certain ``type-loss tradeoff property'' --- which essentially states that agents with higher types are sufficiently more likely to have higher allocations --- then the better of this mechanism or separate Myerson auctions will obtain a constant approximation to the optimal revenue at equilibrium.

In this paper, we demonstrate that many common single-item auction formats satisfy the type-loss tradeoff property, enabling us to apply our framework on a wide variety of mechanisms, addressing many of the aforementioned unturned stones.  Notably, we prove that the first-price auction satisfies the type-loss tradeoff property.  This establishes that randomly selecting between an entry-fee first-price auction and separate reserve-first-price auctions achieves approximately optimal revenue.  This is the first multi-dimensional revenue approximation result in the literature that is based solely on winner-pays-bid mechanisms.

Moreover, other applications of our framework address many of the other problems that come with ex-post entry fee two-part tariff mechanisms.  One such problem comes from the perspective of online learning.  For ex-post tariff mechanisms, entry fees are functions of the opponent's bids and the bidder valuation distributions.  The space of possible mappings from player reports to entry fees is immense, and so when valuation distributions are not known a priori, the online learning of the optimal such mapping can be an arduous task.  In contrast, our entry-fee second price mechanism $\ESP$ exhibits a fixed entry fee. %and is amenable to online learning/tuning.  We can imagine a setting where an auctioneer repeatedly interacts with new members from those distributions, by setting entry fees and then observing the bids and actions (i.e., accept or reject the entry fee) of each bidder.
Assuming that the bidders play the truthful equilibrium, we show that the auctioneer can learn this single parameter in a computationally efficient manner, and achieve average revenue that is a $28$-approximation to the optimal revenue, less a vanishing regret that decays as $\tilde{O}(n\, m\, T^{-1/3})$ after $T$ rounds, with $n$ bidders and $m$ items.\footnote{We note that a slight refinement of our arguments easily establishes a $14$-approximation to the optimal revenue.  This improves upon the previously best-known approximation factor of $32$ achievable via PAC learning techniques~\cite{cai2017learning}, albeit the two results are not directly comparable. Their worse factor accommodates irregular distributions, but we show vanishing regret in the online learning setting.} 

Another problem that comes from ex-post entry fees is the issue of credibility. An auction is credible if the auctioneer cannot benefit by deviating from the prescribed auction rules in a way that is unilaterally undetectable through the auction's communication protocol~\cite{AL18}. When all bids are provided to the mechanism in advance of fees being declared, the auctioneer is able to deceive the bidders.  After observing the bids, she can compute the utilities of the bidders in the auction.  If, for example, some particular bidder had a utility of \$1400, she could charge that bidder a \$1399 entry fee.  From the bidder's perspective, this is suspiciously convenient for the auctioneer, and she might wonder if this entry fee truly was a function purely of the other bids.  However, there is little they can do to investigate this suspicion.  The benefit of posting fixed entry fees in advance is that there is no room for deception here.  And in fact we show, when we instantiate our framework using all-pay auctions (and a minor modification), the resulting entry-fee mechanism ($\GAP$) is \emph{credible} and establishes the first multi-dimensional revenue approximation result for a credible mechanism.

\subsection{Our Techniques: Duality and Type-loss Tradeoffs}  
The main technical contribution of this work lies in the proofs that first-price and other non-truthful auctions satisfy what is called the type-loss tradeoff property.  These arguments rely on a novel geometric technique and offer new insights on the efficiency loss in non-truthful auctions that could be of independent interest.  Roughly, the welfare lost in a single-item first-price or all-pay auction can be achieved as revenue of a posted price mechanism (see \Cref{lem:fp-type-loss} and \Cref{lem:ap-type-loss}).  This property enables us to apply our framework, transforming these single-item auctions into multi-item two-part tariff auctions with good revenue.

The proof of the framework (\Cref{thm:main}) draws heavily from techniques established in~\cite{CDW16}.  In that work, Cai et al.~extend the notion of virtual value to the multi-item setting.  They define it via a partition of the type space that is in terms of the ex post utilities bidders receive in a hypothetical second price auction.  
%The foundation of our research 
Another big contribution of this paper is to redefine this partition in terms of \emph{interim} utility rather than ex post.  The benefit of our technique is that it can be applied to more general auction formats.  This could be, for example, first price or all pay, as well as second price.  This generalization ends up being possible because, while ex post utility is increasing in type for second price auctions specifically, interim utility is increasing in type for all auctions. Moreover, the second key idea is that if we define these regions appropriately in terms of the interim utilities achieved from each item $j$, if that item was to be sold in isolation by some single-item auction $A$ and under some equilibrium $b^j$, then we can upper bound the resulting multi-dimensional virtual value, in terms of the revenue achieved by an entry fee auction, where bidders pay a fixed entry fee in order to participate in a simultaneous $A$ auction for each item. 

To prove this theorem, we perform a decomposition of the multidimensional virtual value in multiple terms, in a manner similar to prior work \cite{CDW16,BILW14}, and bound each of these terms by the revenue of a simple auction. Unlike prior work, when we consider a general auction class $A$ (e.g., first price or all-pay), it comes with the added difficulty of bounding a specific term in this decomposition that relates to the types of bidders that would lose in the hypothetical $A$ auction.  The extent to which we can bound this term for different $A$ is what is described by the ``type-loss tradeoff property'' of auction $A$.  While the type of a bidder who does not win a second price auction is easy to capture through the revenue of a simple auction, the same cannot be said for first price and all pay auctions.  In these auctions, we often have the highest type bidder not winning, and generally, we cannot capture the highest type through revenue.  Our main technical contribution is the proof of a duality whereby both of these worst cases cannot co-exist.  First price and all pay auctions can only misallocate with a high frequency when the expected highest type is attainable through some revenue, and the expected highest type is unattainable when first price achieves almost optimal welfare.

\subsection{Related Work}

%\blcomment{Is there anything more to say / cite about credible mechanisms?}

There has been a recent flurry of results on approximately optimal mechanisms for buyers with multi-dimensional types.  As discussed above, simple constant approximations are known for additive buyers with independent valuations~\cite{BILW14,yao2014n,CDW16}.  For unit-demand buyers, one can likewise obtain a constant approximation to the optimal mechanism with multiple buyers~\cite{chawla2007algorithmic,chawla2010multi,Chawla2015}.  The ideas behind these mechanisms have since been extended to more general valuation classes, including XOS and subadditive valuations~\cite{cai2017simple,rubinstein2018simple,chawla2016mechanism}.  A common theme in many of these mechanisms is the combination of entry fees (or bundle prices, for a single agent) and per-item auctions or prices.

The above line of work focuses on Bayesian incentive compatible mechanisms.  Less is known on approximating optimal revenue with non-truthful auctions.  Hartline, Hoy, and Taggart~\cite{HHT14} develop a framework for bounding the fraction of optimal revenue obtained at equilibrium in various single-item auction formats, such as first-price and all-pay auctions.  Our analysis of the type-loss trade-off for different auction formats shares  inspiration from their equilibrium analysis, as well as from the literature bounding the welfare of equilibria in first-price auctions~\cite{syrgkanis2013composable,hoy2019improved}. 
%\blcomment{Cite literature on uniqueness?}\vscomment{Nah, don't think necessary here.}

Our online learning results for entry-fee mechanisms with second-price auctions relate to a recent literature on the sample complexity of approximately optimal multi-item auctions.  Morgenstern and Roughgarden~\cite{morgenstern2016learning} presented a statistical learning theory approach to bounding the sample complexity of different classes of simple and approximately optimal auctions.  Their approach bounds the pseudo-dimension of different auction classes, but this does not directly imply a polynomial sampling complexity bound for independent additive valuations.  Goldner and Karlin~\cite{goldner2016prior} showed that for bidders with independent additive valuations drawn from regular distributions, one can learn an approximately optimal auction using only a single sample from each bidder's distribution.  Cai and Daskalakis~\cite{cai2017learning} extend this result to non-regular distributions and a broad class of non-additive valuations, by showing that a sequential posted pricing mechanism with entry fees yields approximately optimal revenue and has polynomial sample complexity. All of these works focus on learning from valuation samples in incentive compatible mechanisms. While we likewise restrict our attention to the sample complexity of an incentive compatible mechanism in our framework, ours is an online learning process.  Learning from samples under equilibrium play in non-IC mechanisms is a more subtle task; see Hartline and Taggart~\cite{Hartline2019Sample} for a recent treatment and development in the context of single-parameter types.

Recently and independently, Ferreira and Weinberg~\cite{ferreiraforce} considered the design of credible and incentive compatible single-item auctions.  They show how to design efficient and strategyproof auctions using cryptographic primitives.  We focus on multi-item auctions and show that by relaxing incentive compatibility, it is possible to design non-truthful credible mechanisms without the use of cryptographic primitives.

\section{Mechanism Design Preliminaries and Notation}\label{sec:prelims}

We consider multi-item sealed-bid auctions with $n$ additive bidders and $m$ indivisible items. Each bidder $i$'s valuation/type for each item $j$ is drawn independently from a continuous distribution $D_{ij}$, supported on type space $T_{ij} \subseteq [0, H]$, for some constant $H$ and which admits a continuous bounded density. We will refer to the latter type of distributions as continuous, throughout the paper. The type distributions are common knowledge to the bidders and the auctioneer. The value of a player for a bundle $S$ is the sum of the values for each item in $S$. We will be using the shorthand notation $T_i = \times_j T_{ij}$, $T_{-i} = \times_{i^* \ne i} T_{i^*}$, $T= \times_iT_i$; and analogously we can define $D_i$, $D_{-i}$, $D$.

Each bidder $i$ observes their type $t_i = (t_{i1},\cdots,t_{im})$ and chooses an action $a_i$ (e.g. a bid to submit or a total contingency plan over a multi-round auction).  The auction maps the action profile $a = (a_1,\cdots,a_n)$ to ex-post feasible allocations $x(a) = (x_1(a),\cdots,x_n(a))$ and payments $p^*(a) = (p^*_1(a),\cdots,p^*_n(a))$; where $x_i(a) = (x_{i1}(a),\cdots,x_{im}(a))$ is a vector whose $j$-th entry $x_{ij}(a)\in [0, 1]$ represents the probability of bidder $i$ being allocated item $j$. An allocation is feasible if $\sum_i x_{ij}(a) \leq 1$.
Bidder's have quasi-linear utility: the utility obtained by a bidder is the value they get from the items they receive minus the payment they must make. Thus, the ex-post utility of bidder $i$ is:
$$u^*_{i}(a) = \sum_{j} x_{ij}(a) \cdot t_{ij} - p^*_i(a) = \langle x_i(a), t_i\rangle - p_i^*(a)$$

Since we are also considering non-truthful auctions, we need to define the notion of a Bayes-Nash equilibrium. A \emph{bid strategy} $b = (b_1,\cdots,b_n)$ is a collection of mappings $b_i$ from types $t_i$ to actions $a_i$.\footnote{We restrict attention to pure bid strategies for simplicity. All our results extend to mixed equilibria.}  A bid strategy forms a \emph{pure Bayesian Nash Equilibrium} (BNE), if each bidder has no incentive to deviate conditional on her observed type $t_i$:
$$\EE_{t_{-i} \sim D_{-i}} [u^*_i(b_i(t_i),b_{-i}(t_{-i}))] \geq \EE_{t_{-i} \sim D_{-i}} [u^*_i(a_i',b_{-i}(t_{-i}))] \quad \forall t_i, a_i'$$
Given a BNE $b$, we define the \emph{interim} utilities $u$, allocations $\pi$, and payments $p$ as:
\begin{align*}
    u^b_i(t_i) &= \EE_{t_{-i} \sim D_{-i}} [u^*_i(b_i(t_i),b_{-i}(t_{-i}))] &
    \pi^b_i(t_i) &= \EE_{t_{-i} \sim D_{-i}} [x_i(b_i(t_i),b_{-i}(t_{-i}))]\\
    p^b_i(t_i) &= \EE_{t_{-i} \sim D_{-i}} [p^*_i(b_i(t_i),b_{-i}(t_{-i}))]
\end{align*}
For an auction $A=(x,p^*)$ with bid equilibrium $b$, we define the total expected equilibrium utility, welfare and revenue as:
\begin{align*}
    \Util^b(A) &= \sum_i \EE_{t_i \sim D_i} \ps{u^b_i(t_i)} & 
    \Wel^b(A) &= \sum_i \EE_{t_i \sim D_i} \ps{\langle\pi^b_i(t_i),t_i\rangle} &
    \Rev^b(A) &= \sum_i \EE_{t_i \sim D_i} \ps{p^b_i(t_i)}
\end{align*}
When describing the utility/welfare/revenue of a truthful auction, we will omit the superscript $b$ and assume we are discussing the truthful equilibrium.

\section{Revenue Approximation via Entry-Fee Simultaneous Auctions}

Our goal is to bound the revenue achievable via simultaneous (potentially non-truthful) item-auctions with an entry fee. In this section we will consider a general class of item auctions and define the condition that leads to constant factor revenue guarantees. In the subsequent sections, we will instantiate our analysis to particular item auctions. 

\subsection{Entry-Fee Separate Item Auction Mechanisms}

Let $A$ be an arbitrary single item auction, with allocation and payment rules $(x_A, p_A^*)$. Throughout we will assume that players always have an action in $A$ that gives them zero utility, so that all equilibria $b_A$ of $A$ are \emph{interim individually rational}, i.e. $u_i^{b_A}(t_i)\geq 0$ for all types $t_i$.  The base of our multi-item auction mechanisms will consist of selling the items separately via $A$ auctions.  However, to achieve any reasonable approximation to revenue in multi-dimensional settings, we need to allow for our mechanisms to impose bundle prices. We achieve this by augmenting our simultaneous item auctions with an entry fee. That is, we sell items separately via $A$ auctions, but charge an initial fee $e_i$ to each bidder $i$, that grants them access to the auctions.\\

% Algorithm
\begin{algorithm}[htpb]
	\SetAlgoNoLine
	\KwIn{A single-item auction $A=(x_A, p_A^*)$.}
	\KwIn{For each bidder $i$: an entry fee $e_i$.}
	Each bidder $i$ submits a pair $(z_i, a_i)$ of the decision $z_i\in \{0, 1\}$ to enter the auction or not, and the bid vector $a_i$ to submit to the auctions, if they choose to participate\;
	Let $S=\{i: z_i=1 \text{ or } e_i=0\}$ denote the bidders that enter\;
	For each item $j$, run auction $A$ with bids $a_j=\set{a_{ij} | i \in S}$ to decide allocation $x_A(a_j)$ and payments $p_A^*(a_j)$\;
	For each bidder $i\in S$ return: $x_i(a) = (x_{A, i}(a_1), \ldots, x_{A,i}(a_m))$ and payment $p_i^*(a)=e_i+\sum_{j\in m} p_{A, i}^*(a_j)$\;
	For each bidder $i\notin S$ return zero allocation and payment\;
	\caption{Simultaneous $A$-Item-Auction with Entry Fee ($\EA(e)$)}
	%\label{alg:rand-ea}
\end{algorithm}

For the sake of introducing the concept of entry fees, let's say that each of these separate $A$ auctions admits a unique equilibrium $b=b^1,\cdots,b^m$ induced by valuation distributions $D=\times_j D_j$.  Bidder $i$ will choose to accept the entry fee and participate in the auctions iff her total expected utility exceeds her entry fee; i.e.
\begin{equation}
    \sum_j u^b_{ij}(t_{ij}) \geq e_i
\end{equation}
So, naively, the expected ``entry fee revenue'' (revenue obtained exclusively from the entry fees) will be
\begin{mydef}\label{def:efrefdef}[entry fee revenue]
\begin{equation}
    \efrev^b(A,D,e) = \sum_i e_i\Pr_{t_i \sim D_{i}}\ps{\sum_j u^b_{ij}(t_{ij}) \geq e_i}
\end{equation}
\end{mydef}

%This $\efrev$ term appears in our revenue analysis, and so it is important to devise a mechanism that achieves at least this much revenue.
However, the following caveat makes it not immediately clear that $\EA$ does achieve this entry fee revenue, for every auction $A$.  Bidder $i$'s choice to pay the entry fee and participate will affect the prior the other bidders have on $t_i$.  Defining
\begin{equation}
    T_i^{b,+}(e_i) = \set{t_i \in T_i \middle | \sum_j u^b_{ij}(t_{ij}) \geq e_i}
\end{equation}
and
\begin{equation}
    D_i^{b,+}(e_i) \sim D_i \cdot \mathbbm{1}[t_i \in T_i^{b,+}(e_i)]
\end{equation}
the prior on bidder $i$'s type will now be $D_i^{b,+}(e_i)$.  In the case of the second price auction, this altered prior will not affect the bid equilibrium, as bidding truthfully will still be weakly dominant.  However, when it comes to first price and all-pay, the shifting of the notoriously complex equilibria makes it unclear how entry fee revenue will be affected.  It may seem intuitive that entry fee revenue can only increase when fewer bidders participate, as each bidder's interim utility increases with less competition.  However, the fact that $D_i^{b,+}(e_i)$ does not have item independence induces equilibria in which a bid on one item can be in terms of a type of another, and it is hard to prove anything about the equilibria or interim utilities in this new multi-item mechanism.

This $\efrev$ term appears in our revenue analysis though, and so it is important to devise a mechanism that achieves at least this much revenue.  To resolve this issue, we introduce two modified versions of the $EA$ auction: $\REA$ and $\GEA$. The goal of both are to charge an initial entry fee all the while keeping the equilibria of the subsequent item auctions unchanged.  $\REA$ is a simpler solution to this problem, but $\GEA$ has the unique property that it can be a credible mechanism for certain single-item auctions $A$ without the use of public randomness.  Both are based on the following idea: what if the auctioneer were to submit the bids of all bidders into the item auctions, even those who did not pay the entry fee?  If a bidder who did not enter wins an item, that item is discarded and nobody is charged for it.  At least in this way, the competition faced by each bidder is unaffected by who chooses to enter, and interim utilities are unchanged.

Thus, our goal is to incentivize bidders who are unwilling to pay the entry fee to still report the bids they would have entered had they chosen to participate.  \SCedit{We could simply consider a ``focal equilibrium'' of the entry fee mechanism in which even unwilling bidders report their types truthfully.  The goal of $\REA$ and $\GEA$ though is not to rely on the honesty of non-participating bidders.  $\REA$ solves this problem strongly, admitting a unique equilibrium when the underlying auctions do.\footnote{When the underlying $A$ auctions induce randomized bidding equilibrium, the $\REA$ equilibria will all be equivalent up to marginalization and all achieve good revenue, as is demonstrated in \cref{lem:rand-ea-eq}.}  $\GEA$ may admit additional equilibrium, but none that are affected by the actions of bidders who do not participate.  Thus, a focal equilibrium assumption for $\GEA$ is more realistic in a practical sense.}

The solution of $\REA$ is to not charge the entry fee with a very small probability $\delta > 0$.  That way, even these unwilling bidders will submit equilibrium bids to maximize their utility in the unlikely event that entry fees are not enforced.  And since the entry fee is only voided extremely rarely, the revenue of the mechanism is barely affected.

% Algorithm
\begin{algorithm}[htpb]
	\SetAlgoNoLine
	\KwIn{A single-item auction $A=(x_A, p_A^*)$.}
	\KwIn{For each bidder $i$: an entry fee $e_i$.}
	Each bidder $i$ submits a pair $(z_i, a_i)$ of the decision $z_i\in \{0, 1\}$ to enter the auction or not if an entry fee is imposed, and the bid vector $a_i$ to submit to the auctions, whenever they participate\;
	The auctioneer tosses a biased coin and with probability $\delta$ charges no entry fee (sets $e_i=0$) to the bidders and with probability $1-\delta$, charges entry fees $e=(e_1, \ldots, e_n)$ and filters the bidders based on their submitted $z_i$\;
	Let $S=\{i: z_i=1 \text{ or } e_i=0\}$ denote the bidders that enter\;
	For each item $j$ run auction $A$ with bids $a_j=(a_{1j},\ldots, a_{nj})$ to decide allocation $x_A(a_j)$ and payments $p_A^*(a_j)$\;
	For each bidder $i\in S$ return: $x_i(a) = (x_{A, i}(a_1), \ldots, x_{A,i}(a_m))$ and payment $p_i^*(a)=e_i+\sum_{j\in m} p_{A, i}^*(a_j)$\;
	For each bidder $i\notin S$ return zero allocation and payment\;
	\caption{Simultaneous $A$-Item-Auction with Random Entry Fee ($\mathsf{rand}-\EA(e)$)}
	\label{alg:rand-ea}
\end{algorithm}

As the entry decisions no longer affect the item auctions, we are able to establish the following lemma about equilibrium bidding in $\REA$, which we prove in \Cref{app:rand-ea-eq}:

\begin{lemma}\label{lem:rand-ea-eq}
 Let $\bb$, be any mixed BNE of the $\mathsf{rand}-EA(e)$ auction, with entry fees $e=\{e_i\}$, for type vector distribution $D=\times_{j\in[m]} D_j$, where $D_j=\times_i D_{ij}$. Let $\tilde{\bb}_i^j: T_{ij} \to \Delta(A_{ij})$, denote the marginal action distribution of player $i$ on item $j$ conditional only on her type $t_{ij}$ for item $j$, where $\Delta(A_{ij})$ denotes the set of distributions over bids submitted by player $i$ to the auction for item $j$. Then, $\tilde{\bb}^j = \{\tilde{\bb}_i^j\}$ is a mixed BNE of the $A$ auction for item $j$, when run in isolation. Moreover:
\begin{equation}
    \efrev^{\bb}(\REA(e)) \geq (1-\delta)\efrev^{\tilde{\bb}}(A,D,e)
\end{equation}
where $\efrev^{\bb}(\REA(e))$ represents the revenue of the $\REA$ auction solely coming from the entry fees, and $\efrev^{\tilde{\bb}}(A,D,e)$ is defined in \Cref{def:efrefdef}.
\end{lemma}

The $\GEA$ auction attempts to do something similar without the use of a randomized entry fee.  A mechanism that uses a randomized entry fee in the way of $\REA$ can never be credible 
as the auctioneer would be incentivized to rig the coin flip to whichever outcome generates more revenue.  So, rather than incentivizing low-type bidders to bid logically with a randomized entry fee, $\GEA$ replaces these bidders with "ghost bidders" that submit logical bids in their place.\footnote{
An alternative approach to circumventing this issue is to use a cryptographically-secure source of public randomness to simulate a publicly-observed coin flip.  In this way the auctioneer can charge a randomized entry fee in a trustworthy manner, making $\RFP$ and $\RAP$ credible subject to any computational assumptions that underlie the cryptographic protocol.  We do not take this approach, and in particular $\GEA$ is provably credible without the use of cryptographic primitives.  We thank Ariel Schvartzman, Zachary Langley, Vikrant Ashvinkumar, and Edgar Granados for jointly noting this connection.
%There is a way to circumvent this using ``public randomness''.  Using a cryptographically-secure publicly-oberseved coin flip, the auctioneer can charge the randomized entry fee in a trustworthy manner.  This enables $\RFP$ and $\RAP$ to be credible.  $\GEA$ however is provably credible without the use of cryptographic guarantees.
}
Defining

\begin{equation}
    T_i^{b,-}(e_i) = T_i \setminus T_i^{b,+}(e_i) = \set{t_i \in T_i \middle | \sum_j u^b_{ij}(t_{ij}) < e_i}
\end{equation}
and
\begin{equation}
    D_i^{b,-}(e_i) \sim D_i | t_i \in T_i^{b,-}(e_i)
\end{equation}
$\GEA$ will generate a ghost bidder with type $t_i'$ sampled from distribution $D_i^{b,-}(e_i)$ for every bidder $i$ who does not pay the entry fee.  Then, ghost bids $b(t_i')$ will be injected into the item auctions that, just like in $\REA$, cannot actually lead to an allocation and simply pose as an obstacle for the bidders who do participate.

% Algorithm
\begin{algorithm}[htpb]
	\SetAlgoNoLine
	\KwIn{A single-item auction $A=(x_A, p_A^*)$.}
	\KwIn{For each bidder $i$: an entry fee $e_i$.}
	\KwIn{For each item $j$: an equilibrium bid strategy $b_j$ for an $A$ auction on item $j$}
	Each bidder $i$ submits a pair $(z_i, a_i)$ of the decision $z_i\in \{0, 1\}$ to enter or not the auction and if they decide to enter they also submit a bid vector $a_i$\;
	Let $S=\{i: z_i=1\}$ denote the bidders that decided to enter and $a^S$ the corresponding vector of actions\;
	For each $i\notin S$, draw a type vector $t_i' \sim D_i^{b,-}(e_i)$\;
	For each item $j$ run auction $A$, with bids $\tilde{a}_{j}^S$, such that $\tilde{a}_{ij}^S=a_{ij}^S$ for $i\in S$ and $\tilde{a}_{ij}^S = b_{ij}(t_{ij}')$ for $i\notin S$ to decide allocation $x_A(\tilde{a}_j^S)$ and payments $p_A^*(\tilde{a}_j^S)$\;
	For each bidder $i\in S$ return: $x_i(a^S) = (x_{A, i}(\tilde{a}_1^S), \ldots, x_{A,i}(\tilde{a}_m^S))$ and payment $p_i^*(a^S)=e_i+\sum_{j\in m} p_{A, i}^*(\tilde{a}_j^S)$\;
	For each bidder $i\notin S$ return zero allocation and payment\;
	\caption{Simultaneous $A$-Item-Auction with Entry Fee and Ghost Bidders ($\GEA(e,b)$)}
	\label{alg:one}
\end{algorithm}

If we think of bidder $i$ and the potential ghost that replaces her as a single bidding entity, we note that this entity has type sampled via a distribution that is identical to $D_i$.  We can think of the sampling process as follows: sample a $t_i \sim D_i$, and then if $t_i \in T_i^{b,-}(e_i)$, simply re-roll $t_i$ from that subset.  This means that, if a bid strategy $b$ is an equilibrium in auction $A$, then bidding according to $b$ will constitute a Nash equilibrium in $\GEA$.  However, unlike $\REA$, $\GEA$ may admit additional equilibrium, even when $b^1,\cdots,b^m$ are the unique equilibria in the item auctions.  There may be a circumstance under which a bidding entity has incentive to deviate, but that entity is a ghost.  This ghost is not actually a participant in the auction, merely a fixture of the mechanism, programmed to always bid according to $b$.  As long as only the non-ghost entities have no profitable deviations, the bidding strategies constitute an equilibrium.  Though it is not clear if these ghost-exploiting alternate equilibria exist, we simplify our analysis by only discussing the revenue of $\GEA$ at the equilibrium corresponding to the equilibrium in the separate item auctions.  We refer to this as the ``focal equilibrium'', and at the focal equilibrium it is clear that

\begin{equation}\label{ghost-efrev}
    \efrev^b(\GEA(e,b)) \geq \efrev^b(A,D,e) 
\end{equation}
where $\efrev^{\bb}(\GEA(e))$ represents the revenue of the $\GEA$ auction solely coming from the entry fees

\subsection{Main Theorem}

We will require a crucial, \emph{but non-trivial}, property that the auction $A$ will need to satisfy. This property captures the intuition that high-value bidders at the equilibrium of auction $A$ bid high enough and that losers of the auction are with significant probability not the highest type player. Therefore, the type of the losers can be achieved as revenue by some mechanism. This is a rough intuition of the property, and the formal notion is given below. 

\begin{mydef}[$c$-type-loss trade-off]
Let $A$ be a single-item auction. We say that auction $A$ satisfies the \emph{$c$-type-loss trade-off} property if for any collection of bidders participating in an auction $A$, with any vector of type distributions $D=\times_i D_i$, and for any equilibrium strategy $b$:
$$\EE_{t \sim D}\ps{\max_i t_i \p{1-\pi^b_i(t_i)}} \leq c \cdot \opt(D)$$
where $\opt(D)$ is the optimal revenue in a single-item auction setting with type distributions $D$.
\end{mydef}

We are now ready to state our main theorem.

\begin{theorem}\label{thm:main}
Let $A$ be any single-item auction, satisfying the $c$-type-loss trade-off and which admits an equilibrium strategy $b^j$ for type vector distribution $D_j=\times_i D_{ij}$ that is interim individually rational. Then there exists a set of player-specific entry-fees $e_i$, such that:% the focal equilibrium $b$ of the simultaneous $A$-item-auction with entry fees $e=(e_1, \ldots, e_n)$ and $\{b^j\}$-simulating ghost bidder distribution, $EA(e, D^g(\{b^j\})$ satisfies:
\begin{align}
\opt(D) \leq~& (c+5) \cdot \sum_{j=1}^m \opt(D_j) + 2\cdot \efrev^b(A,D,e) 
\end{align}
where $\opt(D)$ denotes the optimal revenue in the multi-dimensional multi-item auction setting with type distributions $D=\times_i D_i$ and $\opt(D_j)$ is optimal revenue in a single item auction setting with type vector distribution $D_j=\times_i D_{ij}$. 
\end{theorem}

Then, we can establish the following corollaries from \Cref{lem:rand-ea-eq} and \Cref{ghost-efrev} respectively.  First, in the setting where the item auctions admit a unique equilibrium:

\begin{corollary}[$\REA$ with Unique Equilibrium]\label{cor:rand-EA-u}
Let $A$ be any single-item auction, satisfying the $c$-type-loss trade-off and which admits a unique equilibrium strategy $b^j$ for type vector distribution $D_j=\times_i D_{ij}$ that is interim individually rational. Then there exists a set of player-specific entry-fees $e_i$, such that
\begin{align}
\opt(D) \leq~& (c+5) \cdot \sum_{j=1}^m \opt(D_j) + \frac{2}{1-\delta}\cdot \efrev(\REA(e))
\end{align}
\end{corollary}

\begin{corollary}[$\GEA$ at Focal Equilibrium]\label{cor:ghost-EA}
Let $A$ be any single-item auction, satisfying the $c$-type-loss trade-off and which admits a unique equilibrium strategy $b^j$ for type vector distribution $D_j=\times_i D_{ij}$ that is interim individually rational. Then there exists a set of player-specific entry-fees $e_i$, such that at the focal equilibrium of $\GEA$ corresponding to $b$,
\begin{align}
\opt(D) \leq~& (c+5) \cdot \sum_{j=1}^m \opt(D_j) + 2\cdot \efrev^b(\GEA(e,b)) 
\end{align}
\end{corollary}

The $A$ auctions that we consider in the subsequent sections almost always admit unique equilibria.  When we are dealing with item auctions that admit multiple equilibria, however, we need to allow for some additional assumptions regarding equilibrium selection.  The entry fees $e$ that we establish in our proof are in terms of the item auction equilibrium $b$ that the bidders play.  It could be the case that entry fees $e^{(1)}$ achieve good revenue when the subsequent item auctions are played with equilibrium $b^{(1)}$ while entry fees $e^{(2)}$ perform well in the face of equilibrium $b^{(2)}$.  So, $\REA$ could fail to achieve good revenue in the following scenario: when entry fees $e^{(1)}$ are used, the bidders all bid according to equilibrium $b^{(2)}$ and when entry fees $e^{(2)}$ are used, the bidders all bid according to equilibrium $b^{(1)}$.  Such a scenario is utterly absurd, where bidders would actively try to play the equilibrium that lowers their utility to the point where the entry fees are prohibitive, but we must account for it in the statement of our theorem.  Namely, we assume that equilibrium selection is independent of entry fee.

\begin{corollary}[Entry-Fee Oblivious Equilibrium Selection]\label{cor:rand-EA-es}
Let $A$ be any single-item auction, satisfying the $c$-type-loss trade-off and such that all its mixed BNE are interim individually rational. Consider an equilibrium selection process that maps a set of entry fees $e$ to a mixed BNE $\bb^e$ of the $\mathsf{rand}-\EA(e)$ auction. Suppose that the equilibrium selection process satisfies that the marginal bid distribution $\tilde{\bb}_i^j: T_{ij}\to \Delta(A_{ij})$ of each player $i$ for each item $j$ conditional on her type for item $j$, is independent of the entry fee. Then for any such equilibrium selection process, we have:
\begin{align}
\opt(D) \leq~& (c+5) \cdot \sum_{j=1}^m \opt(D_j) + \frac{2}{1-\delta}\cdot \max_{e}\efrev^{\bb^e}(\mathsf{rand}-\EA(e))
\end{align}
\end{corollary}

\paragraph{Remark on $\opt(D_j)$.} The first part of the upper bound corresponds to the sum of the optimal revenues achievable in a set of single-dimensional auction settings. For each item $j$, this optimal quantity $\opt(D_j)$ is achieved by the celebrated Myerson auction \cite{Myerson81} that maps the type $t_{ij}$ of each bidder $i$ for the item to a virtual value $\tilde{\varphi}_{ij}(t_{ij})$ and then allocates the item to the highest virtual value bidder. Moreover, for each of these quantities we can also use existing results on revenue guarantees of truthful and non-truthful simple auctions \cite{HHT14} in single-dimensional settings, to show that this revenue is also achievable by simple, learnable and potentially also credible auctions. For instance, based on the results by \cite{HHT14}, if the type distributions $D_{ij}$ are \emph{regular} (as defined by \cite{Myerson81}), then the first part of the upper bound is approximated to within a constant factor by running a first price auction with a player specific reserve price for each of the items. Similarly, for regular distributions it is also approximated \cite{hartline2009simple} by running a second price auction with player specific reserves.
Thus, bounding the $\opt(D_j)$ term by the revenue of separate reserve auctions and bounding the $\efrev^b(A,D,e)$ by the revenue of the aforementioned entry-fee auctions, our theorem states that \emph{the best of running a separate entry fee auction for each item, or running a grand-bundle entry fee auction, where the entry fee is paid to access all the item auctions}, is a constant factor approximation to the optimal revenue.  For example, in \Cref{sec:FPA}, we prove the following corollary

\begin{cor}\label{cor:first-price}
    Consider a multi-item auction with additive bidders and independent types across bidders $i$ and items $j$, distributed according to $D_{ij}$ and supported in $[0, H]$. Suppose that type distributions $D_{ij}$ are regular and induce unique equilibria in separate first-price auctions for each of the items. Then, for appropriately chosen parameters $r$, $e$, the better of: i) running simultaneous first-price auctions with item and bidder specific reserve prices $\SFP(r)$, ii) running simultaneous first price auctions with bidder specific bundle entry fees $\RFP(e)$, achieves a $\frac{20e-2}{e-1}$-factor approximation to the optimal revenue.
\end{cor}

This is the first multi-dimensional revenue approximation result in the literature that is based solely on winner-pays-bid mechanisms.

\subsection{Proof Outline}

We defer the full proof of the main theorem to \Cref{app:proof-main}, but we outline here the main parts of the proof and some key technical insights.

Our analysis starts with an upper bound on the optimal multi-item auction revenue, as presented in \cite{CDW16}. We adapt this upper bound from the discrete type-space setting to the continuous setting using a discretization argument presented in \Cref{app:proof-disc}.  The proof considers $\epsilon$-discretizations of the continuous type distribution, applies the discrete bound result and then verifies that we can take the limit as $\epsilon$ goes to zero, to get the desired theorem.  The reason why we choose to work with continuous type spaces is primarily because we are interested in analyzing non-truthful auctions, for which there is a plethora of existing equilibrium analysis results (e.g. existence of a monotone pure equilibrium and uniqueness of equilibria) primarily under continuity assumptions on the distribution of types. 

Our revenue upper bound involves a partition of the type space of each player into $m+1$ regions.  These regions are defined in terms of monotone preference functions: for each player $i$ and item $j$, there exists a monotone function of the player's type $t_{ij}$ which assigns a preference score to that item. Then the type vector $t_i$ of player $i$ belongs to partition $j$, roughly if item $j$ is assigned the highest score. More formally:

\begin{mydef}\label{def:MPP}[Monotone Preference Partition of Type Space]
For all $i$, we say that $R_{i, 0},R_{i, 1},\cdots,R_{i, m}$ is a \emph{preference partition} of the type space $T_i$ if it is defined as follows: for each item $j$, there exists \emph{non-decreasing preference functions} $\cU_{i,j}: T_{ij} \to \RR \geq 0$ such that, for all $j \ne 0$
\begin{align*}
    t_i \in R_{i,j} \Leftrightarrow \:  & \cU_{i, j}(t_{ij}) \geq \cU_{i, k}(t_{ik}), \forall k \ne j  ~~~\text{ and } &  \cU_{i, j}(t_{ij}) >& \cU_{i, k}(t_{ik}), \forall k < j ~~~\text{ and } 
    &\cU_{i, j}(t_{ij}) >& 0
\end{align*}
and
$$t_i \in R_{i, 0} \Leftrightarrow \cU_{i, j}(t_{ij}) =0 \quad \forall j$$
i.e. $t_i$ belongs to region $j$ if $\cU_{i, j}(t_{ij})$ is maximal and non-zero, breaking ties lexicographically.
\end{mydef}

%Then we can show the following continuous type analogue of Theorem~31 of \cite{CDW16}. The proof considers $\epsilon$-discretizations of the continuous type distribution, applies the discrete bound result and then verifies that we can take the limit as $\epsilon$ goes to zero, to get the desired theorem. This requires a careful accounting of the discretization errors and showing that the upper bound in the optimal discretized revenue, relates to its continuous analogue, up to an $O(\epsilon)$ error, whenever the partition of the type space is a monotone preference partition.

In general, for any valid monotone preference partitions, we have the following upper bound on the optimal revenue attainable in a multi-item auction.

\begin{lemma}[Revenue Bound via Monotone Preference Partitions of Type Space]\label{lem:contCai}
Consider a multi-item auction setting with additive bidders and independent continuous type distributions $D_{ij}$ on a bounded support $[0, H]$. Let $\{R_{i, j}\}_{i\in [n], j\in [m]}$ be a monotone preference partition of the type space and let $\cF$ denote the space of all interim feasible allocations. Then:
\begin{equation}\label{eqn:rev-bound-cont}
    \opt(D) \leq \sup_{\pi\in \cF} \sum_i \EE_{t_i \sim D_i}\ps{\sum_j \pi_{ij}(t_i)\p{t_{ij}\cdot 1\pb{t_i \not \in R_{i, j}}
+\tilde{\varphi}_{ij}^*(t_{ij})\cdot 1\pb{ t_i \in R_{i, j}}}}
\end{equation}
where $\tilde{\varphi}_{ij}^*(t_{ij})=\max(\tilde{\varphi}_{ij}(t_{ij}),0)$ and $\tilde{\varphi}_{ij}(t_{ij})$ represents Myerson's ironed virtual value function \cite{Myerson81} for the distribution $D_{ij}$.
\end{lemma}

The crucial conceptual contribution of our work is to consider monotone preference partitions of the type space that are described in terms of the interim utilities $u_{ij}^{b^j}(t_{ij})$ of the bidders at some equilibrium $b^j$ for each (potentially non-truthful) item auction $A$ for item $j$. All prior works in the area considered partitions of the type space as a function of ex-post utilities and solely based on the outcomes of a truthful auction $A$ for each item.
In particular, our region definition will assign type $t_i$ to region $j\in \{1, \ldots, m\}$, if item $j$ achieves the highest non-zero interim utility $u_{ij}^{b^j}(t_{ij})$, among all items (and to region $0$ if all interim utilities are zero). By monotonicity of interim utilities of any equilibrium in any single-dimensional mechanism, such a partition of the type space is a monotone preference partition. Hence, we can apply \Cref{lem:contCai}.

Subsequently, we analyze the right-hand-side of Equation~\eqref{eqn:rev-bound-cont} via a decomposition into four terms: $\single$, $\under$, $\over$ and $\surplus$, similar to prior work \cite{BILW14,CDW16}. Our proof shows that this type of analysis can also be carried out even when the regions are defined in terms of interim utilities of non-truthful auctions, still yielding meaningful upper bounds in terms of the revenue of simple multi-item auctions.
The terms $\single$, $\under$ and $\over$ can all be shown to be upper bounded by the sum of the per-item optimal auction revenues; thereby reducing the problem to independent single-dimensional settings. The final term $\surplus$ is shown to be achievable as the revenue of the multi-item $A$ auction with a particular bundle entry fee.

In particular, $\single$ corresponds to the second summand in Equation~\eqref{eqn:rev-bound-cont}, which can be shown to be upper bounded by the sum across items, of the maximum ironed virtual values for each item; which in turn is the optimal per item revenue. 
The first summand on the right-hand-side of Equation~\eqref{eqn:rev-bound-cont} can be divided into the quantities $\under$, $\over$ and $\surplus$. $\under$ corresponds to the part of the event that $t_i$ is not in region $j$ because player $i$ did not bid high enough on item $j$ and hence was not allocated the item. This is exactly where we use the $c$-type-loss trade-off property to show that this quantity, which is roughly the type of the player that lost the item $j$ under equilibrium $b^j$, can be related to the revenue achievable by the optimal auction for item $j$. This property is a non-trivial property of auction $A$ and we will show that it is satisfied by many auctions of interest in the next few sections.

What remains from the first summand, is accounting for the type of the player in the event that player $i$ bids high enough to win auction $j$, but item $j$ is not player $i$'s favorite item as captured by the aforementioned interim utility score. Since in this case, the player received the item, she claims her type for item $j$ as a value, and hence we can relate her type to her utility plus the auctioneer revenue from player $i$ at item $j$. The revenue part is exactly the $\over$ term and it is easily shown to correspond to the revenue achieved by simultaneous $A$-item-auctions without any entry fee. The utility part of this decomposition is the $\surplus$ term, which is much harder to analyze and which roughly corresponds to sums of terms of the form:
\begin{equation}
\EE_{t_{i} \sim D_{i}}\ps{u^b_{ij}(t_{ij}) \cdot 1\ps{\exists k \ne j, u^b_{ik}(t_{ik}) \geq u^b_{ij}(t_{ij})}}
\end{equation}
This term can be shown to be related to the revenue achieved by an simultaneous $A$-item-auction with an entry fee $e_i$, that satisfies that the probability of entry for each player is at least $1/2$. More concretely it satisfies that:
\begin{equation}
    \Pr_{t_i}\ps{\sum_{j} u_{ij}^b(t_{ij}) > e_i} \geq \frac12 
\end{equation}
In fact, we show that it relates to the part of the revenue stemming solely from the collection of entry fees from entrant players. This is where the $\efrev(A,D,e)$ term appears in our upper bound.  The details of this part of the analysis are provided in \Cref{app:surplus}.

\section{Approximately Optimal Fixed Entry-Fee Truthful Auctions}

As a starting point we apply our main theorem to the case where the single item auction $A$ is the second-price (SP) auction; where the highest bidder for an item wins and pays the second highest bid. This will yield a simple and truthful auction that approximates the optimal revenue.  As mentioned previously and as we will demonstrate now, for $SP$ specifically, the unmodified $\ESP$ auction is sufficient for our revenue bounds.  The $\mathsf{rand}$ and $\mathsf{ghost}$ modifications to $\EA$ are unnecessary here as the truthful bidding of $\ESP$ simplifies interim utility analysis.

\SCedit{Recall \Cref{thm:main}, which says that we can bound the revenue of the optimal multi-item auction $\opt(D)$ by
$$\opt(D) \leq (c+5) \cdot \sum_{j=1}^m \opt(D_j) + 2\cdot \efrev^b(A,D,e)$$
as long as the single-item auction $A$ satisfies the $c$-type-loss trade-off property.  We can achieve revenue $\sum_{j=1}^m \opt(D_j)$ by selling items separately via Myerson auctions \cite{Myerson81}.  Thus, if we can
\begin{itemize}
    \item Demonstrate a simple mechanism that obtains revenue at least $\efrev^b(A,D,e)$
    \item Verify the $c$-type-loss trade-off property precondition
\end{itemize}
then we can conclude that simple mechanisms capture a constant factor of the optimal revenue.  We defined
\begin{equation*}
    \efrev^b(A,D,e) = \sum_i e_i\Pr_{t_i \sim D_{i}}\ps{\sum_j u^b_{ij}(t_{ij}) \geq e_i}
\end{equation*}
where $u^b_{ij}(t_{ij})$ is the interim expected utility bidder $i$ obtains in an auction for item $j$ with type $t_{ij}$ facing adversaries with types distributed according to $D_{-i,j}$ and under bidding equilibrium $b$.  We see $e_i\Pr_{t_i \sim D_{i}}\ps{\sum_j u^b_{ij}(t_{ij}) \geq e_i}$ is the revenue achievable by charging bidder $i$ an entry fee $e_i$ to participate in separate $A$ auctions on all the items $j$, in expectation over her type $t_i$.  This entry-fee revenue is achievable by the $\REA$ and $\GEA$ auctions.  In these auctions, $u^b_{ij}(t_{ij})$ will still be the exact interim utility bidder $i$ will receive in the subsequent item $j$ auction if she chooses to enter.  Since she is competing against the bids of even those bidders who choose not to enter, she faces competition of type distribution $D_{-i,j}$ regardless of entry decisions, and $b$ will still be a bidding equilibrium.
\begin{equation}
    \efrev^b(\REA(e))=\efrev^b(\GEA(e))=\efrev^b(A,D,e)
\end{equation}
However, under the unmodified $\EA$ auction, priors will shift based on entry fee decisions.  Recall the definition
\begin{equation}
    T_i^{+}(e_i) = \set{t_i \in T_i \middle | \sum_j u_{ij}(t_{ij}) \geq e_i}
\end{equation}
In the subsequent single-item auctions of this mechanism, the prior distribution on bidder $i$'s type will be $D_i^{+}(e_i) \sim D_i \cdot \mathbbm{1}[t_i \in T_i^{+}(e_i)]$.  That is, we can think of the sampling process of $D_i^{+}(e_i)$ as first drawing a sample $(t_{i1},\cdots,t_{im})$ from $D_i$, and then lowering all the types to $(0,\cdots,0)$ if the sample does not belong to $T_i^{+}(e_i)$.  Intuitively, bidders will have even greater interim utility facing opponents with types distributed according to $D^+$ instead of $D$, as they are facing less competition.  Indeed, it seems almost paradoxical to achieve the goal of high utilities by imposing artificial obstacles and injecting ghost bids.  Our intuition is therefore that unmodified $\EA$ gets strictly better entry-fee revenue than $\REA$ or $\GEA$.  For first-price and all-pay auctions, this intuition is not immediately verifiable as the shifted type distributions will lead to a shifted bidding equilibrium.  For second price though, as truthful bidding will remain the focal equilibrium of the auction, we can verify this seemingly obvious fact.}

\begin{lemma}[$\ESP$ achieves good entry fee revenue]\label{ESP-efrev}
\begin{equation}
    \efrev(\ESP(e)) \geq \efrev(\SP,D,e)
\end{equation}
\end{lemma}
\begin{proof}
% Recall definitions
% \begin{equation}
%     T_i^{+}(e_i) = \set{t_i \in T_i \middle | \sum_j u_{ij}(t_{ij}) \geq e_i}
% \end{equation}
% and $D_i^{+}(e_i) \sim D_i \cdot \mathbbm{1}[t_i \in T_i^{+}(e_i)]$.  That is, we can think of the sampling process of $D_i^{+}(e_i)$ as first drawing a sample $(t_{i1},\cdots,t_{im})$ from $D_i$, and then lowering all the types to $(0,\cdots,0)$ if the sample does not belong to $T_i^{+}(e_i)$.
\SCedit{Slightly abusing notation, in the context of the truthful second-price auction, we replace $u^b_{ij}(t_{ij})$ representing interim utility under bidding equilibrium $b$ with $u^D_{ij}(t_{ij})$ representing interim utility under truthful bidding and type distributions $D$.}  The revenue that the $\ESP$ mechanism extracts from the entry fees $e$ alone is
\begin{equation}
    \efrev(\ESP(e)) = \sum_i e_i\Pr_{t_i \sim D_{i}^+}\ps{\sum_j u^{D^+}_{ij}(t_{ij}) \geq e_i}
\end{equation}
\SCedit{where $u_{ij}^{D^+}(t_{ij})$ denotes the bidder $i$, item $j$ expected interim utility with opponent types distributed according to the $j^{\text{th}}$ marginals of $D^+_{k}(e_k)$ for $k \ne i$.}  We see
\begin{align}\label{eqn:esp-interim-util}
    u^{D^+}_{ij}(t_{ij})&=\EE_{t_{-i}\sim D^+_{-i}}\left[\left[t_{ij} - \max_{k\neq i} t_{kj}\right]_+\right]\geq \EE_{t_{-i}\sim D_{-i}}\left[\left[t_{ij} - \max_{k\neq i} t_{kj}\right]_+\right]=u^{D}_{ij}(t_{ij})
\end{align}
as the $j^{\text{th}}$ marginals of $D$ stochastically dominate those of $D^+$.  Therefore,
\begin{align}
    \efrev(\ESP(e)) &= \sum_i e_i\Pr_{t_i \sim D_{i}^+}\ps{\sum_j u^{D^+}_{ij}(t_{ij}) \geq e_i}\\
    &\geq \sum_i e_i\Pr_{t_i \sim D_{i}}\ps{\sum_j u^{D}_{ij}(t_{ij}) \geq e_i}=\efrev(\SP,D,e)
\end{align}
as desired.
\end{proof}

\SCedit{Thus, all that remains is to show} that the single-item second price auction satisfies the $c$-type-loss trade-off property. Recall that a single-item auction satisfies $c$-type-loss trade-off if, for any type distributions of the bidders $D = \times_iD_i$,
\begin{equation}\label{eqn:tlto}
    \EE_{t \sim D}\ps{\max_i t_i \p{1-\pi_i(t_i)}} \leq c \cdot \opt(D)
\end{equation}
where $\pi_i(t_i)$ is the interim expected allocation of bidder $i$ and $\opt(D)$ is the revenue of the optimal auction (Myerson).  We will show that the second price auction satisfies this for $c=1$. In fact, we will demonstrate a slightly stronger bound, that the left hand side of \eqref{eqn:tlto} is upper bounded by the revenue of the best posted price single-item mechanism $\Rev(\PP)$ (abbreviated $\PP$), which announces some fixed price and allocates to any bidder willing to pay it.  That is,
\begin{equation}
\PP(D) = \max_{r} r\Pr_{t \sim D}\ps{r \leq \max_i t_i} \tag{Posted Price Mechanism Revenue}
\end{equation}

\begin{lemma}[$1$-type-loss trade off of SP]\label{lem:TLTOSP}
In a single-item second-price auction with type vector distribution $D=\times_i D_i$ and under the truthful equilibrium $b$, we have:
$$t_i (1-\pi^b_i(t_i)) \leq \PP(D) \leq \opt(D)$$
for all bidders $i$ and all possible types $t_i$ of bidder $i$.
\end{lemma}
\begin{proof}
    The lemma follows by the following simple set of inequalities:
    \begin{align*}
        t_i (1-\pi^b_i(t_i))
        = t_i \Pr_{t_{-i} \sim D_{-i}}\ps{t_i \leq \max_{j \ne i} t_j}
        \leq \max_r r\Pr_{t_{-i} \sim D_{-i}}\ps{r \leq \max_{j \ne i} t_j}
        \leq \max_{r} r\Pr_{t \sim D}\ps{r \leq \max_i t_i}
    \end{align*}
\end{proof}

Thus we can invoke \Cref{thm:main} to show the following result.

\begin{corollary}\label{thm:esp}
Consider a multi-item auction with additive bidders and independent types across bidders $i$ and items $j$, distributed according to $D_{ij}$ and supported in $[0, H]$. There exists a set of player-specific entry-fees $e=(e_1,\ldots, e_n)$, such that under the truthful equilibrium of the $\ESP(e)$ auction:
\begin{align}
\opt(D) \leq~& 6 \cdot \sum_{j=1}^m \opt(D_j) + 2\cdot \efrev(\ESP(e))
\end{align}
where $\efrev(\ESP(e))$ is the revenue of the $\ESP(e)$ auction solely due to collection of entry fees.
\end{corollary}

Finally, let $\spr(r)$ denote the simultaneous second price auction with reserve prices $r_{ij}$ for each bidder $i$ and item $j$.  In this auction, each item $j$ is sold separately via a second price auction.  If bidder $i$ has the highest bid, and her bid exceeds the player specific reserve $r_{ij}$, she is allocated the item and charged the maximum of $r_{ij}$ and the second highest bid.  Otherwise, the item is not allocated. The results of \cite{hartline2009simple} show that, for regular distributions, this mechanism is a $2$-approximation to $\opt(D_j)$. Thus we can also get the following simplifying corollary:
\begin{corollary}\label{cor:item-bundle-reserve}
Consider a multi-item auction with additive bidders and independent types across bidders $i$ and items $j$, distributed according to $D_{ij}$ and supported in $[0, H]$. Suppose that type distributions $D_{ij}$ are regular. Then, for appropriately chosen parameters $r$, $e$, the better of: i) running simultaneous second price auctions with item and bidder specific reserve prices $\spr(r)$, ii) running simultaneous second price auctions with bidder specific bundle entry fees $\ESP(e)$, achieves a $14$-factor approximation to the optimal revenue.
\end{corollary}

\subsection{Online learnability when prior is unknown to auctioneer}

We conclude this section with a remark on the fact that \Cref{cor:item-bundle-reserve} gives rise to an auction rule that is easy for an auctioneer to optimize in an online manner from historical data, even when the prior distribution of types $D$ is not known to her. We will operate under the assumptions of \Cref{cor:item-bundle-reserve}.  Consider the following online learning setting: at each period $\tau$
\begin{enumerate}
    \item For all $i$, player $i$ draws her type $t_i^{\tau}\sim D_i$ 

    \item The auctioneer posts bidder-specific entry fees $e_i^{\tau}$ and (item, bidder)-specific reserves $r_{ij}^{\tau}$.
    \item Players report bids on all items $b_{ij}^{\tau}$ and their decision $z_i^{\tau}$ to enter in the entry fee mechanism.
    \item The auctioneer flips a coin $\cC^\tau$ and chooses $\spr(r^{\tau})$ with probability $1/2$ and $\ESP(e^{\tau})$ otherwise.
    \item The auctioneer runs the chosen mechanism on the reported input and receives revenue $\cR^{\tau}$
\end{enumerate}
Assuming that players are myopic (or equivalently that each period corresponds to a fresh draw of players from a population), then at each period $\tau$, it is a weakly dominant strategy for all players to report their true types: $b_{ij}^{\tau}=t_{ij}^{\tau}$ and to enter the entry fee mechanism if their belief of their interim utility $u_{i}(t_i^{\tau})=\sum_{j} u_{ij}(t_{ij}^{\tau})$ exceeds the entry fee $e_i^{\tau}$.\footnote{Given that our mechanisms are BIC we still need the players to know $D$ so as to make their entry decision. This is a minimal oracle we need from our bidders to run our auction.}

% Observe that the revenue at each period $\tau$ is an un-biased estimate of the expected revenue under type distribution $D$, of the mechanism that was chosen, i.e.
% \begin{equation}
%     \EE[R^\tau \mid r^{\tau}, e^{\tau}] = \frac{1}{2} \left(\Rev(\spr(r^\tau)) + \Rev(\ESP(e^{\tau}))\right):=f(r^{\tau}, e^{\tau})
% \end{equation}

Each round $\tau \in [1:T]$ and for all $i,j$, the auctioneer will select parameters $r_{ij}^\tau$ and $e_i^\tau$ using $nm+n$ separate instances of the bandit-hedge algorithm, one instance for each parameter \cite{Bubeck2012}.  Each $r_{ij}^\tau$ will be selected from the discrete type space $[0,H]^\epsilon$, representing all of the multiples of $\epsilon$ from $[0,H]$.  Likewise, $e_{i}^\tau$ will be selected from the discrete type space $[0,Hm]^\epsilon$.  At each round, the performance of the chosen arm reported back to the bandit-hedge algorithm for parameters $r_{ij}$ and $e_i$ respectively are
\begin{align*}
    g_{ij}(r^\tau_{ij},\cC^\tau,t^\tau_{*j}) &:= 1\pb{\cC^\tau = \spr}\cdot\max\pb{r^\tau_{ij}, \max_{k\neq i} t^\tau_{kj}}\cdot 1\pb{t^\tau_{ij} \geq \max\pb{r^\tau_{ij}, \max_{k\neq i} t^\tau_{kj}}}\\
    h_{i}(e^\tau_{i},\cC^\tau,e^\tau_{-i},t^\tau_{i*}) &:= 1\pb{\cC^\tau = \ESP}\cdot e^\tau_{i} \cdot 1\pb{\sum_{j=1}^m u^{D^+(e^\tau_{-i})}_{ij}(t^\tau_{ij}) > e^\tau_i} 
\end{align*}

These quantities are observed by the learning auctioneer, as they are in terms of the payments and entry fee decisions of the bidders.  From \cite{Bubeck2012}, the bandit-hedge algorithm guarantees that, for any $\cC^{1:T},t_{*j}^{1:T}$,
\begin{align}
     \max_{r_{ij}^*\in [0, H]^\epsilon} \frac{1}{T} \sum_{\tau=1}^{T} g_{ij}(r^*_{ij},\cC^\tau,t^\tau_{*j}) &\leq \EE_{r_{ij}^{1:T}}\ps{\frac{1}{T}\sum_{\tau=1}^{T} g_{ij}(r^\tau_{ij},\cC^\tau,t^\tau_{*j})}+O\left( H\, \sqrt{\frac{H\log(H/\epsilon)}{\epsilon T}}  + \epsilon\right) \label{eq:olG}
\end{align}
and for any $\cC^{1:T},t_{i*}^{1:T},e_{-i}^{1:T}$,
\begin{align}
     &\max_{e_{i}^*\in [0, Hm]^\epsilon} \frac{1}{T} \sum_{\tau=1}^{T} \nonumber h_{i}(e^*_{i},\cC^\tau,e^\tau_{-i},t^\tau_{i*})\\
     &\leq \EE_{e_{i}^{1:T}}\ps{\frac{1}{T}\sum_{\tau=1}^{T} h_{i}(e^\tau_{i},\cC^\tau,e^\tau_{-i},t^\tau_{i*})}+O\left( H\, \sqrt{\frac{Hm \log(H\,m /\epsilon)}{\epsilon T}}  + \epsilon\right) \label{eq:olH}
\end{align}
taking expectation over the randomness of bandit-hedge.  The reward quantities $g_{ij}$ and $h_i$ have been chosen to satisfy the following.  For any fixed parameters $r^*,e^*$, we have
\begin{align}
    &\sum_{i,j}\EE_{\cC^{1:T},t^{1:T}}\ps{\frac{1}{T}\sum_{\tau =1}^T g_{ij}(r^*_{ij},\cC^\tau,t^\tau_{*j})}=\sum_{i,j}\EE_{\cC,t}\ps{g_{ij}(r^*_{ij},\cC,t_{*j})} \nonumber\\
    &=\EE_\cC\ps{1\set{\cC=\spr}}\sum_{i,j}\EE_{t}\ps{ \max\pb{r^*_{ij}, \max_{k\neq i} t_{kj}}\cdot 1\ps{t_{ij} \geq \max\pb{r^*_{ij}, \max_{k\neq i} t_{kj}}}} \nonumber\\
    &=\frac{1}{2} \Rev(\spr(r^*)) \label{eq:expG}
\end{align}
and for any $e_{-i}^{1:T}$,
\begin{align}
    &\sum_{i}\EE_{\cC^{1:T},t^{1:T}}\ps{\frac{1}{T}\sum_{\tau=1}^T h_{i}(e^*_{i},\cC,e^\tau_{-i},t_{i*})}=\sum_{i}\EE_{\cC,t}\ps{\max_\tau h_{i}(e^*_{i},\cC,e^\tau_{-i},t_{i*})}\nonumber\\
    &=\EE_\cC\ps{1\pb{\cC = \ESP}} \sum_i \EE_t\ps{e^*_{i} \cdot \max_\tau 1\pb{\sum_{j=1}^m u^{D^+(e^\tau_{-i})}_{ij}(t_{ij}) > e^*_i}}\nonumber\\
    &\leq \frac12 \sum_i \EE_t\ps{e^*_{i} \cdot 1\pb{\sum_{j=1}^m u^{D}_{ij}(t_{ij}) > e^*_i}} \label{eqn:interim-learn}\\
    &=\frac{1}{2} \efrev(\SP,D,e^*)\label{eq:expH}
\end{align}
where \eqref{eqn:interim-learn} follows from \eqref{eqn:esp-interim-util}.

\SCedit{From \Cref{thm:main}, \Cref{lem:TLTOSP}, and \cite{hartline2009simple}, we know that there exists bidder-specific entry fees $e^*$ and (item, bidder)-specific reserves $r^*$ such that, for regular type distributions $D$,}
\begin{equation}
    \max_{\substack{r^*_{ij} \in [0,H] \, \forall i,j\\ e^*_i \in [0,Hm]\, \forall i}}\max \set{\Rev(\spr(r^*)),\efrev(\SP,D,e^*)} \geq \frac{1}{14} \opt(D)
\end{equation}
and so
\begin{equation}
    \max_{\substack{r^*_{ij} \in [0,H] \, \forall i,j\\ e^*_i \in [0,Hm]\, \forall i}}\p{\frac12 \p{\Rev(\spr(r^*))+\efrev(\SP,D,e^*)}} \geq \frac{1}{28} \opt(D)
\end{equation}
Even restricting our parameters to multiples of $\epsilon$, we can show 
\begin{equation}\label{eq:eprevbound}
    \max_{\substack{r^{*,\epsilon}_{ij} \in [0,H]^\epsilon \, \forall i,j\\ e^{*,\epsilon}_i \in [0,Hm]^\epsilon\, \forall i}}\p{\frac12 \p{\Rev(\spr(r^{*,\epsilon}))+\efrev(\SP,D,e^{*,\epsilon})}} \geq \frac{1}{28} \opt(D) - \frac{\epsilon}{2}(mn+n)
\end{equation}
Observe that, for any entry fee $e_i^* \in [0, H\cdot m]$, if we consider the largest entry fee $e_i^{*,\epsilon}$ below $e_i^*$ that is a multiple of $\epsilon$, then we have that
$$1\pb{\sum_{j=1}^m u^{D}_{ij}(t_{ij}) > e_i^*} \leq 1\pb{\sum_{j=1}^m u^{D}_{ij}(t_{ij}) > e_i^{*,\epsilon}}$$
Therefore,
$$e_i^*\cdot 1\pb{\sum_{j=1}^m u^{D}_{ij}(t_{ij}) > e_i^*}\leq e_i^{*,\epsilon}\cdot  1\pb{\sum_{j=1}^m u^{D}_{ij}(t_{ij}) > e_i^{*,\epsilon}} + \epsilon$$
and
\begin{align}\label{eq:epr1}
    \efrev(\SP,D,e^*) \leq \efrev(\SP,D,e^{*,\epsilon}) + \epsilon n
\end{align}
Similarly, for every $r_{ij}^* \in [0, H]$, the largest reserve price $r_{ij}^{*,\epsilon}$ below $r_{ij}$ that is a multiple of $\epsilon$ achieves
$$1\ps{t_{ij} \geq \max\pb{r^*_{ij}, \max_{k\neq i} t_{kj}}}\leq 1\ps{t_{ij} \geq \max\pb{r^{*,\epsilon}_{ij}, \max_{k\neq i} t_{kj}}}$$
% Therefore,
% $$\max\pb{r^*_{ij}, \max_{k\neq i} t_{kj}}\cdot 1\ps{t_{ij} \geq \max\pb{r^*_{ij}, \max_{k\neq i} t_{kj}}}\leq \max\pb{r^{*,\epsilon}_{ij}, \max_{k\neq i} t_{kj}}\cdot 1\ps{t_{ij} \geq \max\pb{r^{*,\epsilon}_{ij}, \max_{k\neq i} t_{kj}}} + \epsilon$$
and so
\begin{align}\label{eq:epr2}
    \Rev(\spr(r^{*})) \leq \Rev(\spr(r^{*,\epsilon})) + \epsilon mn
\end{align}
and \eqref{eq:eprevbound} follows from \eqref{eq:epr1} and \eqref{eq:epr2}.  Thus,
\begin{align}
    &\frac{1}{28} \opt(D) - \frac{\epsilon}{2}(mn+n) \nonumber\\
    &\leq \max_{\substack{r^{*,\epsilon}_{ij} \in [0,H]^\epsilon \, \forall i,j\\ e^{*,\epsilon}_i \in [0,Hm]^\epsilon\, \forall i}}\p{\frac12 \p{\Rev(\spr(r^{*,\epsilon}))+\efrev(\SP,D,e^{*,\epsilon})}} \label{eq:fb1}\\
    %&\leq \sum_{i,j}\max_{r^{*,\epsilon}_{ij} \in [0,H]^\epsilon} \EE_{\cC,t}g_{ij}(r^*_{ij},\cC,t_{*j})+\sum_{i}\max_{e^{*,\epsilon}_i \in [0,Hm]^\epsilon}\EE_{\cC,t}h_{i}(e^*_{i},\cC,e^*_{-i},t_{i*})\\
    &\leq \sum_{i,j}\EE_{\cC^{1:T},t^{1:T}}\ps{\max_{r^{*,\epsilon}_{ij} \in [0,H]^\epsilon} \frac{1}{T}\sum_{\tau =1}^T g_{ij}(r^{*,\epsilon}_{ij},\cC^\tau,t^\tau_{*j})} \label{eq:fb2}\\
    &+ \sum_{i}\EE_{\cC^{1:T},t^{1:T}}\ps{\max_{e^{*,\epsilon}_i \in [0,Hm]^\epsilon} \frac{1}{T}\sum_{\tau=1}^T h_{i}(e^{*,\epsilon}_{i},\cC,e^\tau_{-i},t_{i*})} \label{eq:fb3}\\
    &\leq \sum_{i,j}\EE_{\cC^{1:T},t^{1:T},r_{ij}^{1:T}}\ps{\frac{1}{T}\sum_{\tau=1}^{T} g_{ij}(r^\tau_{ij},\cC^\tau,t^\tau_{*j})}+mn \, O\left( H\, \sqrt{\frac{H\log(H/\epsilon)}{\epsilon T}}  + \epsilon\right)\label{eq:fb4}\\
    &+\sum_{i}\EE_{\cC^{1:T},t^{1:T},e_{i}^{1:T}}\ps{\frac{1}{T}\sum_{\tau=1}^{T} h_{i}(e^\tau_{i},\cC^\tau,e^\tau_{-i},t^\tau_{i*})}+n \, O\left( H\, \sqrt{\frac{Hm \log(H\,m /\epsilon)}{\epsilon T}}  + \epsilon\right)\label{eq:fb5}
\end{align}
where \eqref{eq:fb1} follows from \eqref{eq:eprevbound}, \eqref{eq:fb2} and \eqref{eq:fb3} follow from \eqref{eq:expG} and \eqref{eq:expH} respectively, and \eqref{eq:fb4} and \eqref{eq:fb5} follow from \eqref{eq:olG} and \eqref{eq:olH} respectively.  Therefore, setting $\epsilon=HT^{-1/3}$, for $\delta(n,m,H,T)=O\p{\frac{mnH\log(mT)}{T^{1/3}}}$, we have
\begin{align}
    \EE_{\cC^{1:T},t^{1:T},r^{1:T},e^{1:T}}\ps{\frac{1}{T}\sum_{\tau=1}^{T} \cR^\tau}\geq \frac{1}{28} \opt(D) -\delta(n,m,H,T)
\end{align}
as the quantities $g$ and $h$ are components of the revenue $\cR$ obtained by the learner, as desired.

\section{Approximately Optimal First Price Auctions}\label{sec:FPA}

We now move to the case where the auction $A$ is a non-truthful First Price auction ($\FP$); the highest bidder wins and pays her bid. First price single item auctions are known to admit monotone equilibria in our setup with a continuous bounded type distribution with a twice differentiable density \cite{Maskin2000} and under some extra assumptions these equilibria are also unique \cite{Lebrun06} (e.g. if we add any non-zero reserve price). Thus as long as we can show the $c$-type-loss trade-off property for the $\FP$ auction, we can apply \Cref{thm:main}.

\begin{lemma}[Type-Loss Trade-Off for FP]\label{lem:fp-type-loss}
In a single-item first-price auction, with any independent continuous type distribution $D=\times_i D_i$ and under any bid equilibrium $b$, we have
$$\EE_{t \sim D}\ps{\max_i t_i \p{1-\pi^b_i(t_i)}} \leq 4 \PP(D) \leq 4 \opt(D)$$
\end{lemma}

%To state our main corollary we will define the instantiation of the entry fee simultaneous first price auction with ghost bidders as $\EFP(e, b)$, parameterized by a set of entry fees $e_i$ and a set equilibrium strategies $b=(b^1, \ldots, b^m)$, each $b^j$ corresponding to an equilibrium of the single-item first price auction for item $j$. Then the ghost bidders submit a bid on each item drawn based on the equilibrium strategies $b$, conditional on the event that the player decides not to enter (i.e. that the interim utility under $b$ is smaller than the entry fee). Such a mechanism admits the following \emph{focal equilibrium}: player $i$ with type $t_i$ submits bid  $b_i^j(t_{ij})$ on each auction $j$ and decides to enter if the interim utility, i.e. $u_i^b(t_i)=\sum_{j} (t_{ij} - b_i^j(t_{ij}))\, \Pr[b_i^j(t_{ij}) \geq \max_{k\neq i} b_k^j(t_{kj})]$, is greater then $e_i$. 
Then by \Cref{cor:rand-EA-u}:

\begin{cor}\label{FPcor}
    Consider a multi-item auction with additive bidders and independent types across bidders $i$ and items $j$, distributed according to $D_{ij}$ and supported in $[0, H]$. For each item $j$, assume that type profile distributions $D_j=\times_i D_{ij}$ admit a unique equilibrium $b^j$ in a single-item first price auction for item $j$. Then there exists a set of player-specific entry-fees $e=(e_1,\ldots, e_n)$, such that:
    \begin{align}
    \opt(D) \leq~& 9 \cdot \sum_{j=1}^m \opt(D_j) + 2\cdot \efrev(\RFP(e))
    \end{align}
    where $\efrev(\RFP(e))$ is the revenue of the $\RFP(e)$ auction solely due to collection of entry fees.
\end{cor}

Moreover, the results of \cite{HHT14}, show that in a single-item auction settings with independent types and regular distributions $D_j=\times_i D_{ij}$, a first price auction with bidder specific reserves (equal to the monopoly reserve price of each bidder), achieves revenue at least $\frac{e-1}{2e}\opt(D_j)$. Thus if we denote with $\SFP(r)$, the simultaneous version of this auction with item and bidder specific reserves, we have:

\begin{blankcor}
    Consider a multi-item auction with additive bidders and independent types across bidders $i$ and items $j$, distributed according to $D_{ij}$ and supported in $[0, H]$. Suppose that type distributions $D_{ij}$ are regular and induce unique equilibria in separate first-price auctions for each of the items. Then, for appropriately chosen parameters $r$, $e$, the better of: i) running simultaneous first-price auctions with item and bidder specific reserve prices $\SFP(r)$, ii) running simultaneous first price auctions with bidder specific bundle entry fees $\RFP(e)$, achieves a $\frac{20e-2}{e-1}$-factor approximation to the optimal revenue.\footnote{The constant $\frac{20e-2}{e-1}$ arises from substituting $\SFP$ for the separate Myerson auctions $\sum \opt(D_j)$ in \cref{FPcor}, giving a revenue bound of $9\cdot\frac{2e}{e-1}+2$.}
\end{blankcor}

This is the \emph{first multi-dimensional revenue approximation result in the literature that is based solely on winner-pays-bid mechanisms.} The use of first price auction based mechanisms is for instance desirable in settings with multiple competing auctioneers \cite{Renato2020} and many real-world systems rely on first price auction rules \cite{google}. Thus understanding their revenue guarantees is of practical importance.

Applying \Cref{cor:ghost-EA} instead of \Cref{cor:rand-EA-u} will give a revenue bound that replaces $\efrev(\RFP(e))$ with $\efrev(\GFP(e,b))$ (at the focal equilibrium).  This is a step towards achieving a multi-dimensional revenue bound that is credible.  However, $\GFP$ is still not credible in the formal sense defined in \cite{AL18}: whenever a ghost bidder wins, the auctioneer has incentive to deviate, without the bidders noticing, and allocate the item to an entrant bidder. In the next \Cref{sec:all-pay}, we show how this problem can be fixed by switching to all-pay auctions.

\subsection{Proof of \Cref{lem:fp-type-loss}: Type-Loss Trade-Off for FPA}\label{sec:fp-type-loss}

First we note that by the fact that utilities are quasi-linear:
\begin{align*}
    t_i\, (1-\pi^b_i(t_i))\leq t_i -\p{\pi^b_i(t_i)\, t_i-p^b_i(t_i)} = t_i-u^b_i(t_i)
\end{align*}
Thus it suffices to show that:
\begin{equation}
    \EE_{t \sim D}\ps{\max_i\p{t_i-u^b_i(t_i)}} \leq 4 \PP(D)
\end{equation}

Let $B_i^{b}(t_{-i}) = \max_{k\neq i} b_k(t_k)$ denote the highest other bid in the $\FP$ auction as a function of the type profile of player $i$'s opponents. Then by the rules of the $\FP$ auction and the BNE condition:
\begin{align*}
    u^b_i(t_i) \geq \max_{r\leq t_i}\, (t_i-r) \Pr_{t_{-i} \sim D_{-i}}\ps{B^b_i(t_{-i}) < r}
\end{align*}
As a first step we show a structural lemma that connects a player's interim equilibrium utility with her type and the distribution of the highest other bid. 
\begin{lemma}[Box Lemma]\label{boxes}
    Let $F: \RR>0 \to [0,1]$ be any function and let $t >0$.  Consider the quantities $u(t) = \max_{b} (t-b)\, F(b)$
    %, that is the area of the largest box that lies underneath the $F$ curve with bottom right corner at $(t,0)$.  Let 
    and $a(t) = \max_{r \leq t} r\, (1-F(r))$.
    %, that is the area of the largest box that lies above the $F$ curve with top left corner at $(0,1)$.  
    Then, $\sqrt{u(t)}+\sqrt{a(t)} \geq \sqrt{t}$.
\end{lemma}
\begin{proof}
By definition of $u(t), a(t)$, we must have $F(x) \leq \overline{F}(x):=\frac{u(t)}{t-x}$ and $F(x) \geq \underline{F}(x):= 1-\frac{a(t)}{x}$ for all $x \in [0,t]$ (see \Cref{fig:boxes}). Thus, we must have $\overline{F}(x)\geq \underline{F}(x)$ for all $x\in [0, t]$, i.e.:
\begin{equation}\label{eqn:curve-gap}
    \min_{x\in [0, t]} \p{\overline{F}(x) - \underline{F}(x)} = \min_{x\in [0, t]} \p{\frac{u(t)}{t-x} + \frac{a(t)}{x}} - 1 \geq 0
\end{equation}
Observe that the function $\frac{u(t)}{t-x} + \frac{a(t)}{x}$ is convex in $x$, when $x\in [0, t]$. Hence, by writing down the first order condition and solving for $x$, we find that it is minimized at: $x = \frac{t\sqrt{a(t)}}{\sqrt{u(t)}+\sqrt{a(t)}}$,
yielding:
\begin{align*}
    \min_{x\in [0, t]} \p{\overline{F}(x) - \underline{F}(x)}
    =~& \frac{u(t)}{t}\cdot\frac{\sqrt{u(t)}+\sqrt{a(t)}}{\sqrt{u(t)}}+\frac{a(t)}{t}\cdot\frac{\sqrt{u(t)}+\sqrt{a(t)}}{\sqrt{a(t)}}-1
    = \frac{\p{\sqrt{u(t)}+\sqrt{a(t)}}^2}{t}-1
\end{align*}
Thus for Equation~\eqref{eqn:curve-gap} to hold it must be that $\sqrt{u(t)}+\sqrt{a(t)} \geq \sqrt{t}$, as desired.
\end{proof}
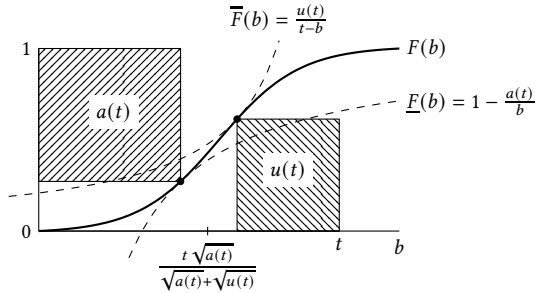
\begin{figure}[htpb]
\begin{center}
\begin{tikzpicture}[scale=0.8, every node/.style={scale=0.8}]
\draw[-] (0,0) -- (5,0) node[below] {$t$};
\draw[-] (0,0) -- (6,0) node[below] {$b$};
\draw[-] (0,0) node[left] {$0$} -- (0,3) node[left] {$1$};
\draw[domain=0:6,smooth,variable=\x,thick] plot ({\x}, { 3*(1/(1+exp(-1.5*(\x-3))) - .011)/.978}) node[right] {$F(b)$};
\draw[domain=-.5:4,smooth,variable=\x,dashed] plot ({\x}, {3.13/(5-\x)}) node[above] {$\overline{F}(b) = \frac{u(t)}{t-b}$};
\draw[domain=1.5:6,smooth,variable=\x,dashed] plot ({\x}, {3 - 5.16/\x}) node[right] {$\underline{F}(b) = 1 - \frac{a(t)}{b}$};
\draw[pattern=north west lines] (3.3,0) rectangle (5,1.84);
\draw[pattern=north east lines] (0,0.816) rectangle (2.36,3);
\node[fill=white] at (4.15,1) {$u(t)$};
\node[fill=white] at (1.25,1.95) {$a(t)$};
\draw[] (2.81, -.07) node[below]  {$\frac{t\, \sqrt{a(t)}}{\sqrt{a(t)} + \sqrt{u(t)}}$} -- (2.81, .07) ;
\filldraw (2.36, 0.816) circle (1.5pt);
\filldraw (3.3, 1.84) circle (1.5pt);
\end{tikzpicture}
\end{center}
\caption{Pictorial representation of quantities in Box Lemma \Cref{boxes} and its proof.}\label{fig:boxes}
\end{figure}

Applying this lemma with $F(r) := \Pr_{t_{-i}\sim D_{-i}}[B^b_i(t_{-i}) < r]$, i.e. the CDF of $B^b_i(t_{-i})$, gives:
$$\sqrt{u_i^{b}(t_i)} \geq \sqrt{t_i} - \sqrt{a_i(t_i)}$$
where $a_i(t_i) = \max_{r \leq t_i} r\Pr_{t_{-i} \sim D_{-i}}[r \leq B^b_i(t_{-i})]$. Since by definition $a_i(t_i) \leq t_i$, we can square the last inequality to obtain:
$$u^b_i(t_i) \geq \p{\sqrt{t_i}-\sqrt{a_i(t_i)}}^2 \geq t_i - 2\sqrt{a_i(t_i) \cdot t_i}$$
Moreover,
\begin{align*}
    a_i(t_i) &= \max_{r \leq t_i} r\Pr_{t_{-i} \sim D_{-i}}[r \leq B^b_i(t_{-i})]
    \leq \max_{r \leq t_i} r\Pr_{t_{-i} \sim D_{-i}}[r \leq \max_{j \ne i} t_{j}]\leq \max_{r} r\Pr_{t \sim D}[r \leq \max_i t_i] = \PP(D)
\end{align*}
Thus,
\begin{align*}
    \EE_{t \sim D}\ps{\max_i\p{t_i-u^b_i(t_i)}}
    \leq \EE_{t \sim D}\ps{\max_i\p{2\sqrt{a_i(t_i) \cdot t_i}}}
    \leq 2\sqrt{\PP(D)}\cdot \EE_{t \sim D}\ps{\sqrt{\max_i t_i}}
\end{align*}
We finish with the following lemma that shows that the expected root of the highest type can be achieved to within a constant factor as the root of the revenue of the best posted price mechanism, i.e. $\EE_{t \sim D}\ps{\sqrt{\max_i t_i}}\leq 2\sqrt{\PP(D)}$. Combined with the above inequality, this would conclude the overall proof of the lemma.

\begin{lemma}[Root Lemma]\label{sqrt} For a single item auction setting with any type profile distribution $D$:
    $$\EE_{t \sim D}\ps{\sqrt{\max_i t_i}} \leq 2\sqrt{\PP(D)}$$
\end{lemma}
\begin{proof}
By the definition of $\PP(D)$:
$$\Pr\ps{\sqrt{\max_i t_i} \leq x} = \Pr\ps{\max_i t_i \leq x^2} \geq \max\pb{1-\frac{\PP(D)}{x^2},0}$$
since $\PP(D) \geq x^2\Pr_{t \sim D}\ps{x^2 \leq \max_i t_i} $.
Hence:
\begin{align*}
    \EE_{t \sim D}\ps{\sqrt{\max_i t_i}} =~& \int_0^\infty \p{1- \Pr\ps{\sqrt{\max_i t_i} \leq x}} dx
    \leq \int_0^\infty \p{1- \max\p{1-\frac{\PP(D)}{x^2},0}} dx\\
    =~& \int_0^{\sqrt{\PP(D)}} 1dx + \int_{\sqrt{\PP(D)}}^\infty \frac{\PP(D)}{x^2} dx
    = 2\sqrt{\PP(D)}
\end{align*}
\end{proof}

\section{Approximately Optimal Credible Auction}\label{sec:all-pay}
We finally discuss the case where the auction $A$ is a non-truthful All-Pay auction (AP); the highest bidder wins and every bidder pays her bid. In our setting, with continuous type distributions and a common interval support of $[0, H]$, all-pay auctions admit pure monotone equilibria (see e.g. \cite{amann1996asymmetric,Lebrun06,luo2014optimal}). 

Crucially we show that all-pay auctions also satisfy the $c$-type-loss trade-off property. In fact we show a much more general statement: all sealed high-bid-wins auctions, where players do not overbid, and the auctioneer charges at most the player's bid (irrespective of allocation), satisfy that property. All-pay auctions certainly meet these criteria.
%\footnote{Monotonicity follows easily for all-pay auctions: by Myerson's characterization theorem \cite{Myerson81} at any pure BNE equilibrium $b$, interim allocation is monotone and: $b_i(t)=t \cdot \pi_i(t) - \int_{0}^{t} \pi_i(z) dz \implies b_i'(t)=t \cdot \pi_i'(t)\geq 0$.} (see e.g. \cite{amann1996asymmetric,Lebrun06,luo2014optimal}).
%\snote{Say something about all-pay equilibria; Is it okay to claim uniqueness assuming that the IEEE paper is correct? I am a bit worried about this because I feel that the proof in their appendix is not rigorous enough.} \cdcomment{Say ``regular'' somewhere....}\vscomment{I added the reference. We can verify it more concretely later. It seemed convincing to me. Either way we just say it in passing.}

\begin{lemma}[Type-Loss Trade-Off for General Auctions]\label{lem:ap-type-loss}
Consider any sealed-bid single-item auction, where the highest bidder wins and irrespective of allocation is charged at most her bid. Moreover, suppose that bidders do not bid more than their type at equilibrium. Then for any independent type distribution $D=\times_i D_i$ and under any no-overbidding bid equilibrium $b$, we have
$$\EE_{t \sim D}\ps{\max_i t_i \p{1-\pi^b_i(t_i)}} \leq 4 \PP(D) \leq 4 \opt(D)$$
\end{lemma}

We defer the proof of \Cref{lem:ap-type-loss} to Section~\Cref{sec:all-pay-type-loss}. To prove the lemma, we observe that the equilibrium interim utility in any such auction is at least the largest box below the highest-other-bidder CDF curve, minus the largest box above that curve. Then we can follow similar analysis as in the first price auction to handle the first part of this decomposition and carry over an extra ``largest box above curve'' term; which subsequently is upper bounded by the best posted price revenue.

To state our main corollary we will define the instantiation of the entry fee simultaneous all-pay auction with ghost bidders as $\GAP(e, b)$, parameterized by a set of entry fees $e_i$ and a set equilibrium strategies $b=(b^1, \ldots, b^m)$, each $b^j$ corresponding to an equilibrium of the single-item all-pay auction for item $j$. Then the ghost bidders submit a bid on each item drawn based on the equilibrium strategies $b$, conditional on the event that the player decides not to enter (i.e. that the interim utility under $b$ is smaller than the entry fee). Such a mechanism admits the following \emph{focal equilibrium}: player $i$ with type $t_i$ submits bid  $b_i^j(t_{ij})$ on each auction $j$ and decides to enter if the interim utility, i.e. $u_i^b(t_i)=\sum_{j} t_{ij} \cdot \Pr[b_i^j(t_{ij}) \geq \max_{k\neq i} b_k^j(t_{kj})] - b_i^j(t_{ij}) $, is greater then $e_i$. Then by \Cref{cor:ghost-EA}:

\begin{cor}
Consider a multi-item auction with additive bidders and independent types across bidders $i$ and items $j$, distributed according to $D_{ij}$ and supported in $[0, H]$. For each item $j$, let $b^j$ denote an equilibrium of the single-item all-pay auction with type profile distribution $D_j=\times_i D_{ij}$. Then there exists a set of player-specific entry-fees $e=(e_1,\ldots, e_n)$, such that in the focal equilibrium $b$ of the $\GAP(e, b)$:
    \begin{align}
    \opt(D) \leq~& 9 \cdot \sum_{j=1}^m \opt(D_j) + 2\cdot \efrev^b(\GAP(e, b))
    \end{align}
    where $\efrev(\GAP(e, b))$ is the revenue of the $\GAP(e, b)$ auction solely due to collection of entry fees.
\end{cor}

Combining the latter with the results of \cite{HHT14}, we have:

\begin{cor}\label{allpay}
 Consider a multi-item auction with additive bidders and independent types across bidders $i$ and items $j$, distributed according to $D_{ij}$ and supported in $[0, H]$. Suppose that type distributions $D_{ij}$ are regular. Then, for appropriately chosen parameters $r$, $e$, the better of: i) running simultaneous first-price auctions with item and bidder specific reserve prices $\SFP(r)$, ii) running simultaneous all-pay auctions with bidder specific bundle entry fees $\GAP(e, b)$ (at the focal equilibrium), achieves a $\frac{20e-2}{e-1}$-factor approximation to the optimal revenue.
\end{cor}

This is the \textit{first multi-dimensional revenue approximation result in the literature with a credible mechanism}. The results of \cite{AL18} show that $\SFP(r)$, for any setting of the parameter $r$, is a credible mechanism. In the subsequent section, we prove that $\GAP(e,b)$ is also credible, for any setting of the parameter $e$ and under any bid equilibrium $b$.

\subsection{Credibility of entry-fee all pay auction} \label{credibility}

In this section, we formally define the criteria for a mechanism to be credible, as in \cite{AL18}, and then prove that $\GAP$ is a credible mechanism. We view a mechanism as a communication game between the auctioneer and the bidders. Let $n$ denote the number of bidders, $X$ denote the set of all possible outcomes of the mechanism and $T=\times_i T_i$ denote the type-space of the bidders. The auctioneer is viewed as a player $0$ in the auction with utility (revenue) denoted by
\[u_0:X\times T\mapsto \mathbb{R}.
\]
The bidders are viewed as players indexed by the set $[n]$. At each step, player $0$ contacts a player $i\in [n]$ privately. It sends a message and receives a reply. At any step, player $0$ can choose an outcome $x\in X$ and end the game. Each player $i$ may have access to a part of the outcome. Let $S_i$ denote the strategy of player $i$. Let $o_j(S_0,S_1,...,S_n,t)$ denote the observation of player $j$ when the auctioneer plays $S_0$, bidder $i$ plays $S_i$ and the type profile is $t=(t_1,t_2,\dots,t_n)$. The observation $o_j$ includes the set of all messages received by player $j$ along with the part of the outcome it observes.

\begin{mydef}
Given a promised strategy profile $(S_0,S_1,\dots,S_n)$, we define an auctioneer strategy $\hat{S_0}$ to be \textbf{\emph{safe}} if for every player $i\in [n]$ and type profile $t=(t_1,t_2,\dots,t_n)$, there exists $\hat{t}_{-i}$ such that
\[
o_i(\hat{S_0},S_1,\dots,S_n,t) = o_i(S_0,S_1,\dots,S_n,(t_i,\hat{t}_{-i})),
\]i.e., even if the auctioneer deviates from the promised strategy, there is an equivalent \emph{innocent} explanation for each bidder's observation.
\end{mydef}

Let $S_0^*(S_0,S_1,\dots,S_n)$ denote the set of all safe strategies for the auctioneer. The auctioneer is restricted to play only a strategy $S\in S_0^*$ the messaging game. This is a reasonable constraint because if the auctioneer plays a strategy that is not "safe", the deviation can be easily detected by some bidder $i$.

\begin{mydef}
A mechanism with strategy profile $(S_0^G,S_1,\dots,S_n)$ is credible if
\[
S_0^G \in {\arg\max}_{S_0\in S_0^*(S_0^G,S_1,\dots,S_n)} \mathbb{E}_t\left[u_0 (S_0,S_1,\dots,S_n,t)\right].
\]
\end{mydef}
%EAP is credible

\begin{theorem}[Credibility of $\GAP$]
The ghost entry-fee all pay ($\GAP$) auction is a credible mechanism.
\end{theorem}

\begin{proof}
The communication protocol for the $\GAP$ mechanism proceeds as follows.  Recall that the individual entry fees are fixed and known in advance to all players.  Each bidder first sends a message to the auctioneer stating whether they will pay the entry fee.  Those that do then provide bids to the auctioneer for each of the separate all-pay auctions.  The auctioneer then stops the game and returns an outcome.  In the auctioneer's promised strategy, this is done by simulating the bids of any bidders who chose not to pay their entry fees and then choosing an outcome consistent with the all-pay auction evaluated on each item separately.

Let $(S_0, S_1, \dotsc, S_n)$ be the promised (non-deviating) strategy profile for the $\GAP$ mechanism.  As the auctioneer's only decision point is the selection of the outcome, strategies differ only in this choice of outcome.  However, note that each bidder's payment under the promised strategy $S_0$ is a deterministic function of their action, since entry fees are fixed and each bidder's payment in an all-pay auction is determined by their bid.  Thus all safe strategies $\hat{S}_0 \in S^*_0$ must agree on each agent's payment, and therefore generate the same revenue for the auctioneer.\footnote{Note that since we assume the goods are consumed immediately and have $0$ value for the seller, the seller's payoff is entirely determined by their revenue.}  We conclude that $S_0$ weakly maximizes revenue over all safe strategies, and hence the $\GAP$ mechanism is credible.
%
%Once the bidders accept and pay the entry fee, they essentially play separate all-pay auctions for each item. \cite{Akbarpour} showed that single-dimensional all-pay auctions is credible. Therefore, all we need to show is that the auctioneer does not have any incentive to deviate when the item is allocated to a ghost bidder. This is true because in an all-pay auction, the auctioneer is indifferent to who the item is allocated to, under the assumption that the item to be sold is time-sensitive and unallocated item (item allocated to a ghost bidder) cannot be resold. Every safe deviation that allocates the object involves charging every participating bidder her bid and since each item is time-sensitive and discarded after the auction, if not allocated, no safe deviation yields a higher revenue and hence, EAP is credible. 
%\vscomment{I don't understand this last sentence!! What is a time-sensitive item?} 
%\vscomment{I added some sentences in between from the other argument. Can we check to make sure we are happy with this argument.}
\end{proof}

\subsection{Proof of \Cref{lem:ap-type-loss}: Type Loss Trade-Off for General Auctions}\label{sec:all-pay-type-loss}

Consider any sealed high-bid-wins auction $A$ that charges each player at most their bid (irrespective of winning or losing), i.e. $p_{A, i}^*(b)\leq b_i$, and suppose the players do not overbid at equilibrium, i.e. $b_i(t_i)\leq t_i$. 
We first note that by the fact that utilities are quasi-linear:
$$t_i\, (1-\pi^b_i(t_i))\leq t_i -\p{\pi^b_i(t_i)\, t_i-p^b_i(t_i)} = t_i-u^b_i(t_i)$$
Therefore, it suffices to show that
$$\EE_{t \sim D}\ps{\max_i\p{t_i-u^b_i(t_i)}} \leq 4\, \PP(D)$$

Let $B^b_i(t_{-i})=\max_{k\neq i} b_k(t_k)$, denote the highest other equilibrium bid and let $F_i$, denote the CDF of this random variable over the randomness of $t_{-i} \sim D_{-i}$, i.e. $F_i(r) = \Pr_{t_{-i}\sim D_{-i}}\ps{B^b_i(t_{-i}) < r}$. 
By the assumptions on the allocation and payment rule of $A$, the best-response equilibrium condition and the fact that bidders do not overbid, we have:
\begin{align}
    u_i^b(t_i) \geq~& \max_{r\leq t_i} \p{t_i\cdot F_i(r) - \EE_{t_{-i}\sim D_{-i}}\ps{p_{A,i}^*\p{r, b_{-i}(t_{-i})}}} \nonumber\\
    \geq~& \max_{r\leq t_i} \p{t_i\cdot F_i(r) - r} \tag{by assumption that payment is at most bid} \\
    =~& \max_{r\leq t_i} \p{(t_i - r)\cdot F_i(r) - r\cdot (1-F_i(r))} \nonumber\\
    \geq~& \underbrace{\max_{r\leq t_i} (t_i - r)\cdot F_i(r)}_{u_i(t_i)} - \underbrace{\max_{r\leq t_i} r\cdot (1-F_i(r))}_{a_i(t_i)} \label{eqn:general-A-util-bound}
\end{align}

Applying \Cref{boxes} with $F$ being the CDF of $B^b_i(t_{-i})$, gives:
\begin{equation*}
    \sqrt{u_i(t_i)} \geq \sqrt{t_i} - \sqrt{a_i(t_i)}
\end{equation*}
where we note that $a_i(t_i) = \max_{r \leq t_i} r\Pr_{t_{-i} \sim D_{-i}}[r \leq B^b_i(t_{-i})]$. Since by definition $a_i(t_i) \leq t_i$, we can square the last inequality to obtain:
\begin{equation*}
u_i(t_i) \geq \p{\sqrt{t_i}-\sqrt{a_i(t_i)}}^2 = t_i - 2\sqrt{a_i(t_i) \cdot t_i} + a_i(t_i)
\end{equation*}
Combining with Equation~\eqref{eqn:general-A-util-bound}:
\begin{equation}
    u_i^b(t_i) \geq u_i(t_i) - a_i(t_i) \geq t_i - 2\sqrt{a_i(t_i) \cdot t_i}
\end{equation}
Moreover, since by assumption players do not overbid at equilibrium:
\begin{align*}
    a_i(t_i) &= \max_{r \leq t_i} r\Pr_{t_{-i} \sim D_{-i}}[r \leq B^b_i(t_{-i})]
    \leq \max_{r \leq t_i} r\Pr_{t_{-i} \sim D_{-i}}[r \leq \max_{j \ne i} t_{j}]\leq \max_{r} r\Pr_{t \sim D}[r \leq \max_i t_i] = \PP(D)
\end{align*}
Thus,
\begin{align*}
    \EE_{t \sim D}\ps{\max_i\p{t_i-u^b_i(t_i)}}
    \leq \EE_{t \sim D}\ps{\max_i\p{2\sqrt{a_i(t_i) \cdot t_i}}}
    \leq~& 2\sqrt{\PP(D)}\cdot \EE_{t \sim D}\ps{\sqrt{\max_i t_i}}\\
    \leq~& 4\, \PP(D) \tag{by \Cref{sqrt}}\\
\end{align*}

%\section*{Acknowledgements}
%We thank Ariel Schvartzman for noting the connection to cryptographic protocols for public randomness.

%\input{position_auctions}

%\appendix

\section{Proof of Theorem~\ref{thm:main}}\label{app:proof-main}

\begin{blankthm}
Let $A$ be any single-item auction, satisfying the $c$-type-loss trade-off and which admits an equilibrium strategy $b_j$ for type vector distribution $D_j=\times_i D_{ij}$ that is interim individually rational. Then there exists a set of player-specific entry-fees $e_i$, such that:% the focal equilibrium $b$ of the simultaneous $A$-item-auction with entry fees $e=(e_1, \ldots, e_n)$ and $\{b^j\}$-simulating ghost bidder distribution, $EA(e, D^g(\{b^j\})$ satisfies:
\begin{align}
\opt(D) \leq~& (c+5) \cdot \sum_{j=1}^m \opt(D_j) + 2\cdot \efrev^b(A,D,e) 
\end{align}
where $\opt(D)$ denotes the optimal revenue in the multi-dimensional multi-item auction setting with type distributions $D=\times_i D_i$, $\opt(D_j)$ is optimal revenue in a single item auction setting with type vector distribution $D_j=\times_i D_{ij}$, and
\begin{equation}
    \efrev^b(A,D,e) = \sum_i e_i\Pr_{t_i \sim D_{i}}\ps{\sum_j u^b_{ij}(t_{ij}) \geq e_i}
\end{equation}
\end{blankthm}

Our starting point is \Cref{lem:contCai}, whose proof is provided in \Cref{app:proof-disc}. The lemma provides an upper bound on the revenue of the optimal multi-item mechanism that holds generally for all valid ``monotone preference partitions'' $\{R_{i, j}\}_{i\in [n], j\in [m]}$.  Recall that monotone preference partitions are partitions of the type spaces of the bidders defined in terms of non-decreasing preference functions $\cU_{i,j}: T_{ij} \to \RR \geq 0$

\begin{blankdef}[Monotone Preference Partition of Type Space]
For all $i$, we say that $R_{i, 0},R_{i, 1},\cdots,R_{i, m}$ is a \emph{preference partition} of the type space $T_i$ if it is defined as follows: for each item $j$, there exists \emph{non-decreasing preference functions} $\cU_{i,j}: T_{ij} \to \RR \geq 0$ such that, for all $j \ne 0$
\begin{align*}
    t_i \in R_{i,j} \Leftrightarrow \:  & \cU_{i, j}(t_{ij}) \geq \cU_{i, k}(t_{ik}), \forall k \ne j  ~~~\text{ and } &  \cU_{i, j}(t_{ij}) >& \cU_{i, k}(t_{ik}), \forall k < j ~~~\text{ and } 
    &\cU_{i, j}(t_{ij}) >& 0
\end{align*}
and
$$t_i \in R_{i, 0} \Leftrightarrow \cU_{i, j}(t_{ij}) =0 \quad \forall j$$
i.e. the preference function assigns an index to each item that is a monotone function of that item's type $t_{ij}$ and then the type vector $t_i$ belongs to the region $j$ with the highest positive index, breaking ties lexicographically, or to region $0$ if all indices are zero.

\end{blankdef}

Then, \Cref{lem:contCai} states that, for any monotone preference partition $\{R_{i, j}\}_{i\in [n], j\in [m]}$ of the type spaces, letting $\cF$ denote the set of interim feasible allocations, we have:
\begin{equation}\label{eqn:upper-bound}
\opt(D) \leq \sup_{\pi \in \cF} \underbrace{\sum_i \EE_{t_i \sim D_i}\ps{\sum_j \pi_{ij}(t_i)\p{t_{ij}\cdot 1\pb{ t_i \not \in R_{i, j}}
+\tilde{\varphi}_{ij}^*(t_{ij})\cdot 1\pb{ t_i \in R_{i, j}}}}}_{\vw(\pi)}
\end{equation}
where $\tilde{\varphi}_{ij}^*(t_{ij})=\max(\tilde{\varphi}_{ij}(t_{ij}),0)$ and $\tilde{\varphi}_{ij}(t_{ij})$ represents Myerson's ironed virtual value function \cite{Myerson81} for the distribution $D_{ij}$. We will refer to the latter bound as the multi-dimensional virtual welfare ($\vw$). 

To apply this bound it suffices to define the monotone preference functions $\cU_{i,j}: T_{ij} \to \RR \geq 0$. We define this preference function in terms of the interim utility $u_{ij}^{b^j}(t_{ij})$ that bidder $i$ receives in auction $j$ with type vector distribution $D_j=\times_i D_{ij}$ under equilibrium $b^j$. For simplicity of notation we will denote $u_{ij}^{b^j}(t_{ij})$ with $u_{ij}^{b}(t_{ij})$ moving forward. We will define the preference function as:
\begin{equation}
    \cU_{i, j}(t_{ij}) = u^b_{ij}(t_{ij})
\end{equation}

Observe that this preference function is non-decreasing in $t_{ij}$, since interim utility at any equilibrium of any mechanism in a single dimensional environment is non-decreasing in type, by standard results in single-dimensional mechanism design (see e.g. Theorem 2.2 of \cite{hartline2013mechanism}).

% \paragraph{Description of regions} Intuitively, we define our preference regions in terms of the best interim utility item under strategy profile $b$: each type vector $t_i$ of bidder $i$ induces a ranking of the items based on interim utility: $j_1,j_2,\cdots,j_m$ such that
% \begin{equation}
% u^{b}_{ij_1}(t_{ij_1}) \geq u^{b}_{ij_2}(t_{ij_2}) \geq \cdots \geq u^{b}_{ij_m}(t_{ij_m}) \geq  0
% \end{equation}
% breaking ties lexicographically. Then, we assign the bidder to the highest ranked item for her type, with non-zero interim utility. Thus we can re-express the partitions as:
% \begin{equation}\label{eqn:region-char}
% \begin{aligned}
%     t_i \in R_{i, j} \Leftrightarrow  &\text{ } u_{ij}^b(t_{ij})>0\\
%     &\text{ and for all $k$: } u^b_{ik}(t_{ik}) \leq  u^b_{ij}(t_{ij})\\
%     &\text{ and for all $k < j$: } u^b_{ik}(t_{ik}) < u^b_{ij}(t_{ij})
% \end{aligned}
% \end{equation}
% We say $t_i \in R_0$ if it belongs to no other regions.

\subsection{Decomposition of Upper Bound $\vw$}
Applying \Cref{lem:contCai} on the aforementioned monotone preference partition regions, we now further decompose the right hand side of Equation~\eqref{eqn:upper-bound} into four terms: $\under$, $\single$, $\over$, $\surplus$ that we will subsequently bound separately.

First, observe that by the characterization of regions
\begin{align}\label{eqn:bound-1}
    1\pb{ t_i \not \in R_{i,j}} 
    %&\leq \Pr_{t_{-i} \sim D_{-i}} \ps{\begin{aligned} &b_{ij}(t_{ij}) \leq B^b_{ij}(v_{-i})\\
    % \vee \: &\exists k \ne j, \text{ s.t. } u^b_{ik}(t_{ik}) \geq u^b_{ij}(t_{ij}) \text{ and } b_{ik}(t_{ik}) > B^b_{ik}(v_{-i})
    %\end{aligned}}\\
    &\leq 1\pb{\exists k \ne j, \text{ s.t. } u^b_{ik}(t_{ik}) \geq u^b_{ij}(t_{ij})}
\end{align}
disregarding the lexicographic tie-breaking.  Let $\cZ_{ij}^b(t_i)$ denote the event that \emph{item $j$ is a strictly favorite item for player $i$ in terms of interim utility under equilibrium $b$}, i.e.:
\begin{equation}
    \cZ_{ij}^b(t_i) := \{\forall k \ne j: u^b_{ij}(t_{ij}) > u^b_{ik}(t_{ik})\}
\end{equation}
and let $\bar{\cZ}_{ij}^b(t_i)$ denote its complement. Thus:
\begin{align*}
    1\pb{ t_i \not \in R_{i,j}} \leq~& 1\pb{\bar{\cZ}_{ij}^b(t_i)} \leq 
    \underbrace{1\pb{\bar{\cZ}_{ij}^b(t_i)}\cdot \pi_{ij}^b(t_{ij})}_{\text{Prob allocated $j$, but $j$ not strict favorite}} + 
    \underbrace{(1 - \pi_{ij}^b(t_{ij}))}_{\text{Prob not allocated item $j$ in auction $A$}}
\end{align*}
where we remind that $\pi_{ij}^b(t_{ij})$ is the interim allocation of player $i$ in a single item auction $A$ for item $j$ under equilibrium $b^j$. So, we can upper bound and decompose the virtual welfare $\vw(\pi)$ as
\begin{align*}
\vw(\pi) \leq~& \sum_i \EE_{t_i \sim D_i}\ps{\sum_j\pi_{ij}(t_i) \cdot t_{ij} \cdot 1\pb{\bar{\cZ}_{ij}^b(t_i)} \cdot \pi_{ij}^b(t_{ij})}
\tag{$\nonfav(\pi)$}\\
&+ \sum_i \EE_{t_i \sim D_i}\ps{\sum_j\pi_{ij}(t_i)\cdot t_{ij} \cdot (1 - \pi_{ij}^b(t_{ij})} \tag{$\under(\pi)$} \\
&+\sum_i \EE_{t_i \sim D_i}\ps{\sum_j\pi_{ij}(t_i) \cdot \tilde{\varphi}^*_{ij}(t_{ij}) \cdot 1\pb{ t_i \in R_{i, j}}} \tag{$\single(\pi)$}
\end{align*}
We further decompose $\nonfav$ by invoking the quasi-linearity of player utilities in each item auction, i.e. $t_{ij} \cdot \pi_{ij}^b(t_{ij}) = u_{ij}^b(t_{ij}) + p_{ij}^b(t_{ij})$:
\begin{align*}
    \nonfav(\pi) =& \sum_i \EE_{t_i \sim D_i}\ps{\sum_j\pi_{ij}(t_i)\cdot 1\pb{\bar{\cZ}_{ij}^b(t_i)}\cdot\p{u^b_{ij}(t_{ij})+p^b_{ij}(t_{ij})}}\\
    \leq & \underbrace{\sum_i \EE_{t_i \sim D_i}\ps{\sum_j\pi_{ij}(t_i)\cdot p^b_{ij}(t_{ij})}}_{\over(\pi)}
    + \underbrace{\sum_i \EE_{t_i \sim D_i}\ps{\sum_j\pi_{ij}(t_i)\cdot u^b_{ij}(t_{ij})\cdot 1\pb{\bar{\cZ}_{ij}^b(t_i)}}}_{\surplus(\pi)}
\end{align*}
Which completes our final upper bound decomposition as:
\begin{equation}
    \vw(\pi) \leq \over(\pi) + \surplus(\pi) + \under(\pi) + \single(\pi)
\end{equation}
In the next sections we will prove the following bounds, which complete the proof of our theorem.
\begin{align*}
    \single(\pi) &\leq \sum_{j=1}^m \opt(D_j) &
    \under(\pi) &\leq c\cdot \sum_{j=1}^m \opt(D_j)\\
    \over(\pi) &\leq \sum_{j=1}^m \opt(D_j) &
    \surplus(\pi) &\leq 3\sum_{j=1}^m \opt(D_j) + 2\cdot \efrev^b(A,D,e)
\end{align*}

\subsection{Upper Bounding $\single$}\label{app:single}

Since $\pi_{ij}$ is an interim feasible allocation, we have that there exists an ex-post feasible allocation $x_{ij}$, such that $\pi_{ij}(t_i)=\EE_{t_{-i}\sim D_{-i}}\left[x_{ij}(t)\right]$. Invoking this fact and the fact that $\varphi_{ij}^*(t_{ij})\geq 0$, we have:
\begin{align*}
    \single(\pi) =~& \sum_i \EE_{t_i \sim D_i}\ps{\sum_j\pi_{ij}(t_i) \cdot \tilde{\varphi}^*_{ij}(t_{ij})\cdot 1\pb{ t_i \in R_{i,j}}}\\
    \leq~& \sum_i \EE_{t_i \sim D_i}\ps{\sum_j\pi_{ij}(t_i)\cdot \tilde{\varphi}^*_{ij}(t_{ij})} \tag{since $\tilde{\varphi}^*_{ij}(t_{ij})\geq 0$}\\
    =~& \sum_i \EE_{t_i \sim D_i}\ps{\sum_j \EE_{t_{-i}\sim D_{-i}}\ps{x_{ij}(t_i,t_{-i})}\cdot \tilde{\varphi}^*_{ij}(t_{ij})} \tag{by interim feasibility of $\pi_{ij}$}\\
    =~& \sum_j \EE_{t \sim D} \ps{\sum_i x_{ij}(t)\cdot \tilde{\varphi}^*_{ij}(t_{ij})}\\
    \leq~& \sum_j \EE_{t \sim D} \ps{\max_i \tilde{\varphi}^*_{ij}(t_{ij})} \tag{by ex-post feasibility of $x_{ij}$}\\
    =~& \sum_{j} \opt(D_j) \tag{by Myerson's \cite{Myerson81} theorem}
\end{align*}

\subsection{Upper Bounding $\under$}\label{app:under}

We rearrange $\under(\pi)$ to be in terms of the ex-post feasible allocation $x$ that gives rise to interim allocation $\pi$.
\begin{align*}
    \under(\pi) =~& \sum_i \EE_{t_i \sim D_i}\ps{\sum_j\pi_{ij}(t_i)\cdot t_{ij}\cdot \p{1-\pi^b_{ij}(t_{ij})}}\\
    =~&\sum_i \EE_{t_i \sim D_i}\ps{\sum_j\EE_{t_{-i}\sim D_{-i}}\ps{x_{ij}(t_i,t_{-i})}t_{ij} \p{1-\pi^b_{ij}(t_{ij})}}\\
    =~& \sum_j \EE_{t \sim D} \ps{\sum_i x_{ij}(t)\cdot t_{ij}\cdot \p{1-\pi^b_{ij}(t_{ij})}}\\
    \leq~& \sum_j \EE_{t \sim D} \ps{\max_i t_{ij} \p{1-\pi^b_{ij}(t_{ij})}} \tag{by ex-post feasibility of $x_{ij}$}\\
    \leq~& \sum_j c \cdot \opt(D_j) \tag{by $c$-type-loss trade off property of $A$}
\end{align*}

\subsection{Upper Bounding $\over$}\label{app:over}

\begin{align*}
\over(\pi) =~& \sum_i \EE_{t_i \sim D_i}\ps{\sum_j\pi_{ij}(t_i)\cdot p^b_{ij}(t_{ij})}\\
\leq& \sum_i \EE_{t_i \sim D_i}\ps{\sum_j p^b_{ij}(t_{ij})} \tag{by interim feasibility: $\pi_{ij}(t_{ij})\leq 1$}\\
=& \sum_j \Rev^{b^j}(A)\\
\leq& \sum_{j} \opt(D_j)
\end{align*}
where $\Rev^{b_j}(A)$ represents the revenue of the auction $A$ on item $j$ under equilibrium $b^j$ with type vector distribution $D_j$.

\subsection{Upper Bounding $\surplus$}\label{app:surplus}

By rearranging the terms in $\surplus$ and invoking the fact that types are independent across items, we have:
\begin{align*}
    \surplus(\pi) =~& \sum_i \EE_{t_i \sim D_i}\ps{\sum_j\pi_{ij}(t_i)\cdot u^b_{ij}(t_{ij})\cdot 1\pb{\bar{\cZ}_{ij}^b(t_i)}}\\
    \leq & \sum_i \EE_{t_i \sim D_i}\ps{\sum_ju^b_{ij}(t_{ij})\cdot 1\pb{\bar{\cZ}_{ij}^b(t_i)}} \tag{by interim feasibility: $\pi_{ij}(t_{ij})\leq 1$}\\
    = & \sum_i \sum_j\EE_{t_{ij} \sim D_{ij}}\ps{u^b_{ij}(t_{ij})\cdot \EE_{t_{i,-j} \sim D_{i,-j}} \ps{1\pb{\bar{\cZ}_{ij}^b(t_i)}}} \tag{ independence across items}\\
    = & \sum_{i, j} \EE_{t_{ij} \sim D_{ij}}\ps{u^b_{ij}(t_{ij})\Pr_{t_{i,-j} \sim D_{i,-j}} \ps{\exists k \ne j, u^b_{ik}(t_{ik}) \geq u^b_{ij}(t_{ij})}} \tag{definition of $\bar{\cZ}_{ij}^b(t_i)$}
\end{align*}

Analyzing the relative size of each $u^b_{ij}(t_{ij})$ will be fundamental to bounding the surplus term.  Intuitively, in the event that $u^b_{ij}(t_{ij})$ is not too large, its contribution to the $\surplus$ sum will be not too large and therefore boundable.  When $u^b_{ij}(t_{ij})$ is very large, bounding $\surplus$ will still be possible due to the fact that the probability there exists an even larger $u^b_{ik}(t_{ik})$ will be small. Thus, we will analyze $\surplus$ by splitting into an analysis of these two regimes, denoted as $\core$ and $\tail$.  The pivotal point that defines these two regimes is based on an interim utility threshold $r^b_{i}$ defined as follows.
\begin{align*}
    r^b_{ij}&=\max_{x}\p{x \Pr_{t_{ij} \sim D_{ij}}[u^b_{ij}(t_{ij}) \geq x]}\\
    r^b_i &= \sum_j r^b_{ij}
\end{align*}
We decompose $\surplus$ based on this interim utility threshold:
\begin{align*}
    \surplus(\pi) \leq~& \sum_{i, j}\EE_{t_{ij} \sim D_{ij}}\ps{u^b_{ij}(t_{ij})\Pr_{t_{i,-j} \sim D_{i,-j}} \ps{\exists k \ne j, u^b_{ik}(t_{ik}) \geq u^b_{ij}(t_{ij})} \cdot 1[u^b_{ij}(t_{ij}) \geq r^b_i]} \tag{$\tail$}\\
    &+ \sum_{i, j} \EE_{t_{ij} \sim D_{ij}}\ps{u^b_{ij}(t_{ij}) \cdot 1[u^b_{ij}(t_{ij}) < r^b_i]} \tag{$\core$}
\end{align*}

\paragraph{Upper bounding $\tail$} We upper bound this term by $r^b = \sum_i r^b_i$.  At the end of the analysis, we prove that $r^b \leq \sum_{j} \opt(D_j)$.  First, by union bound
$$\Pr_{t_{i,-j} \sim D_{i,-j}} \ps{\exists k \ne j, u^b_{ik}(t_{ik}) \geq u^b_{ij}(t_{ij})} \leq \sum_{k \ne j} \Pr_{t_{ik} \sim D_{ik}} \ps{u^b_{ik}(t_{ik}) \geq u^b_{ij}(t_{ij})}$$
By the definition of $r^b_{ik}$, we have that:
\begin{equation}
r^b_{ik} \geq u^b_{ij}(t_{ij}) \Pr_{t_{ik} \sim D_{ik}}\ps{u^b_{ik}(t_{ik}) \geq u^b_{ij}(t_{ij})}    
\end{equation}
and so we can bound $\tail$ as: 
\begin{align*}
    \tail \leq~& \sum_{i, j}\EE_{t_{ij} \sim D_{ij}}\ps{1\ps{u^b_{ij}(t_{ij}) \geq r^b_i} \sum_{k \ne j}u^b_{ij}(t_{ij})\Pr_{t_{ik} \sim D_{ik}} \ps{u^b_{ik}(t_{ik}) \geq u^b_{ij}(t_{ij})}}\\
    \leq &\sum_{i, j} \EE_{t_{ij} \sim D_{ij}}\ps{1\ps{u^b_{ij}(t_{ij}) \geq r^b_i} \sum_{k \ne j}r^b_{ik}}\\
    \leq &\sum_{i, j} r^b_i \cdot \EE_{t_{ij} \sim D_{ij}}\ps{1\ps{u^b_{ij}(t_{ij}) \geq r^b_i}}\\
    = &\sum_{i, j} r^b_i \cdot \Pr_{t_{ij} \sim D_{ij}}\ps{u^b_{ij}(t_{ij}) \geq r^b_i} \\
    \leq &\sum_{i, j} r^b_{ij} \tag{since $r^b_{ij} \geq x \Pr[u^b_{ij}\geq x]$ for all $x$}\\
    =~& \sum_{i} r_i^b = r^b
\end{align*}

\paragraph{Upper bounding $\core$}
%We show $EA \geq \frac{\text{CORE}}{2}-r^b$, which gives the bound $2r^b+2EA \geq \text{CORE}$.  
For notational convenience, let
\begin{equation}
c^b_{ij}(t_{ij})=u^b_{ij}(t_{ij})\cdot 1\ps{u^b_{ij}(t_{ij})<r^b_{i}}
\end{equation}
so that:
\begin{equation}
\core = \sum_{i, j} \EE_{t_{ij} \sim D_{ij}}\ps{c^b_{ij}(t_{ij})}
\end{equation}
Now, we consider the $\efrev^b(A,D,e)$ term with entry fee $e_i$ for bidder $i$ defined as:
\begin{equation}
e^b_{i}= \left[\p{\sum_{j} \EE_{t_{ij}} \ps{c^b_{ij}(t_{ij})}} - 2r^b_{i}\right]_+
\end{equation}
where $[x]_+ :=\max\{x, 0\}$.
This will be a valid entry fee when used with the $\REA$ and $\GEA$ auctions as it is a non-negative constant that only depends on the type distributions $D_{ij}$, and is not specific to any $t_{ij}$.

We will show that each bidder $i$ accepts the entry fee with probability at least $1/2$.  Bidder $i$ accepts the entry fee iff her total interim utility over the auctions exceeds the fee.  Thus, if we can show 
\begin{equation}
    \Pr_{t_i \sim D_{i}}\ps{\sum_j u^b_{ij}(t_{ij}) \geq e^b_i} \geq 1/2
\end{equation}
then we know% the expected revenue of $EA(e, D^g)$ (in equilibrium $b$ as described in the definition of the Theorem) from entry fees alone is at least
\begin{align}
    \efrev^b(A,D,e) &= \sum_i e_i\Pr_{t_i \sim D_{i}}\ps{\sum_j u^b_{ij}(t_{ij}) \geq e_i}\\
    &\geq \frac12 \sum_i e^b_i \geq \frac12 \sum_i\p{\sum_{j} \EE_{t_{ij}} [c^b_{ij}(t_{ij})] - 2r^b_{i}} = \frac{\core}{2}-r^b
\end{align}
this would imply that:
\begin{equation}
\core \leq 2\,r^b + 2\,\efrev^b(A,D,e)  
\end{equation}
as desired.  We make use of the following lemma, originally proved in \cite{BILW14},

\begin{lemma}
    Let $x$ be a positive single dimensional random variable drawn from $F$ of finite support, such that for any number $a$, $a \cdot \Pr_{x \sim F} [x \geq a] \leq \cB$ where $\cB$ is an absolute constant. Then, for any positive number $s$, the second moment of the random variable $x_s = x \cdot 1[x \leq s]$ is upper bounded by $2\cdot \cB \cdot s$.
\end{lemma}

%PROOF

Applying this lemma with $x=u^b_{ij}(t_{ij})$, $\cB = r^b_{ij}$ and $s = r^b_i$, we obtain:
\begin{equation}
\EE\ps{\p{c^b_{ij}(t_{ij})}^2} \leq 2r^b_ir^b_{ij}
\end{equation}
Since $c^b_{ij}(t_{ij})$ are independent across items,
\begin{equation}
    \text{Var}\ps{\sum_j c^b_{ij}(t_{ij})} = \sum_j \text{Var}\ps{c^b_{ij}(t_{ij})} \leq \sum_j\EE\ps{\p{c^b_{ij}(t_{ij})}^2} \leq 2\, \p{r^b_i}^2
\end{equation}
By Chebyshev, we know:
\begin{equation}
\Pr_{t_i \sim D_i}\ps{\sum_j c^b_{ij}(t_{ij}) \leq \sum_j \EE[c^b_{ij}(t_{ij})] - 2r^b_i}\leq \frac{\text{Var}\ps{\sum_j c^b_{ij}(t_{ij})}}{4\, \p{r^b_i}^2} \leq \frac12
\end{equation}
Moreover, since $c_{ij}^b(t_{ij})$ are non-negative, we have that in the case where the $[\cdot]_+$ binds and $e_i=0$, then it is definitely true that:
\begin{equation}
\Pr_{t_i \sim D_i}\ps{\sum_j c^b_{ij}(t_{ij}) \leq \ps{\sum_j \EE[c^b_{ij}(t_{ij})] - 2r^b_i}_+} \leq \frac12
\end{equation}
Hence, $\Pr_{t_i \sim D_i}\ps{\sum_j c^b_{ij}(t_{ij}) \leq e_i} \leq \frac12$. Since, we also have that
$u^b_{ij}(t_{ij})\geq c^b_{ij}(t_{ij})$, we can conclude that $\Pr\ps{\sum_j u^b_{ij} > e^b_i} \geq 1/2$, as desired.

\paragraph{Upper bounding $r^b$} We conclude the proof of the bound on $\surplus$ by providing an upper bound on $r^b$. 
We can obtain revenue $r^b$ via selling the items separately, where each item is sold via an entry-fee $A$ auction solely for that item. More concretely, for each item $j$, each bidder $i$ can choose to pay an entry fee $e_{ij}$ to access an auction $A$ on item $j$.  Bidder $i$ can choose whether or not to buy into the item $j$ auction totally independently of her choice for the other auctions.  We will again be using ghost bidders for all bidders who do not pay the entry fee.  Thus, bidder $i$'s utility for entering the item $j$ auction is $u^b_{ij}(t_{ij})$, and she will pay the entry fee iff $u^b_{ij}(t_{ij}) \geq e_{ij}$. The maximum entry fee revenue we can obtain in such an auction is equal to:
\begin{equation*}
    \max_{e_{ij}} e_{ij} \Pr_{t_{ij} \sim D_{ij}} [u^b_{ij}(t_{ij}) \geq e_{ij}] = r^b_{ij}
\end{equation*}
Thus, setting entry fees optimally on all items for all bidders, we obtain entry fee revenue $\sum_i \sum_j r^b_{ij}=r^b$.  The revenue obtained from these separate entry-fee $A$ auctions on each item is upper bounded by the revenue obtained from separate optimal single item auctions on each item,  giving 
\begin{equation}
    r^b\leq \sum_{j} \opt(D_j)
\end{equation}
as desired.

\paragraph{Concluding} Combining all the above analysis, we have:
\begin{align*}
    \surplus(\pi) \leq~& \tail + \core \leq r^b + \p{2\, r^b + 2\, \efrev^b(A,D,e)}\\
    \leq~& 3 \sum_j \opt(D_j) + 2\, \efrev^b(A,D,e) 
\end{align*}

\section{Proof of Lemma~\ref{lem:contCai}}\label{app:proof-disc}

\begin{blanklemma}[Revenue Bound via Monotone Preference Partitions of Type Space]
Consider a multi-item auction setting with additive bidders and independent continuous type distributions $D_{ij}$ on a bounded support $[0, H]$. Let $\{R_{i, j}\}_{i\in [n], j\in [m]}$ be a monotone preference partition of the type space and let $\cF$ denote the space of all interim feasible allocations. Then:
\begin{equation}
    \opt(D) \leq \sup_{\pi\in \cF} \sum_i \EE_{t_i \sim D_i}\ps{\sum_j \pi_{ij}(t_i)\p{t_{ij}\cdot 1\pb{t_i \not \in R_{i, j}}
+\tilde{\varphi}_{ij}^*(t_{ij})\cdot 1\pb{ t_i \in R_{i, j}}}}
\end{equation}
where $\tilde{\varphi}_{ij}^*(t_{ij})=\max(\tilde{\varphi}_{ij}(t_{ij}),0)$ and $\tilde{\varphi}_{ij}(t_{ij})$ represents Myerson's ironed virtual value function \cite{Myerson81} for the distribution $D_{ij}$.
\end{blanklemma}

% Recalling our definition of a ``Monotone Preference Partition of Type Space'':\\
% For all $i$ and $t_{-i} \in T_{-i}$, we say that $R_{i, 0}^{t_{-i}},R_{i, 1}^{t_{-i}},\cdots,R_{i, m}^{t_{-i}}$ is a \emph{preference partition} of the type space $T_i$ if it is defined as follows: for each item $j$ and each $t_{-i}$, there exists \emph{non-decreasing preference functions} $\cU_{i,j}^{t_{-i}}: T_{ij} \to \RR \geq 0$ such that, for all $j \ne 0$
% \begin{align*}
%     t_i \in R_{i,j}^{t_{-i}} \Leftrightarrow \:  & \cU_{i, j}^{t_{-i}}(t_{ij}) \geq \cU_{i, k}^{t_{-i}}(t_{ik}), \forall k \ne j  ~~~\text{ and } &  \cU_{i, j}^{t_{-i}}(t_{ij}) >& \cU_{i, k}^{t_{-i}}(t_{ik}), \forall k < j ~~~\text{ and } 
%     &\cU_{i, j}^{t_{-i}}(t_{ij}) >& 0
% \end{align*}
% and
% $$t_i \in R_{i, 0}^{t_{-i}} \Leftrightarrow \cU_{i, j}^{t_{-i}}(t_{ij}) =0 \quad \forall j$$
% i.e. the preference function assigns an index to each item that is a monotone function of that item's type $t_{ij}$ and then the type vector $t_i$ belongs to the region $j$ with the highest positive index, breaking ties lexicographically, or to region $0$ if all indices are zero.\\

% Recalling our definition of a ``feasible interim allocation'' $\pi^M$:\\
% For all $i,j$, there exists some mapping $x_{ij}^M: \RR^{n \times m} \to [0,1]$ so that $\pi_{ij}^M(t_i) = \EE_{t_{-i}\sim D_{-i}}\ps{x_{ij}(t_i,t_{-i})}$ and $\sum_i x_{ij}(t) \leq 1$ for all items $j$ and $t \in T$.  We note that $\pi^M$ doesn't have to be the allocation of any specific mechanism.  It just must satisfy the feasibility constraint.\\

Our starting point is the following lemma of \cite{CDW16} that applies to discrete types and discrete type distributions. 
We will subsequently provide a discretization argument that allows us to prove the continuous analogue of it presented in \Cref{lem:contCai}.

\begin{theorem}[Theorem 31 \cite{CDW16}]\label{CaiThm}
Consider a multi-item auction setting with discrete type space $T^+$ and discrete valuation distribution $D^+$.  For each $v_{-i} \in T_{-i}^+$, let $R_0^{v_{-i}},R_1^{v_{-i}},\cdots,R_m^{v_{-i}}$ be a partition of the type space $T_i^+$ into ``upwards-closed'' regions. That is, for all $j \ne 0$,
$$t_i = (t_{i1},\cdots,t_{ij},\cdots,t_{im}) \in R_{i,j}^{v_{-i}} \Rightarrow (t_{i1},\cdots,t_{ij}^*,\cdots,t_{im}) \in R_{i,j}^{v_{-i}} \text{ for all }t_{ij}^*>t_{ij}$$

Let $M$ be any BIC mechanism with values drawn from $D^+$ that has interim allocation and payment $\pi^M,p^M$ in the truthful equilibrium.  The expected revenue of $M$ in the truthful equilibrium is upper bounded by the expected virtual welfare of the same allocation rule with respect to the canonical virtual value function $\varphi_i$. In particular,
\begin{align}
    \Rev(M) \leq \sum_{i, j} \EE_{t_i \sim D_i^+}\ps{\pi^M_{ij}(t_i)\p{t_{ij}\cdot \Pr_{v_{-i} \sim D_{-i}^+} \ps{ t_i \not \in R_{i,j}^{v_{-i}}}
+\tilde{\varphi}_{ij}(t_{ij})\cdot \Pr_{v_{-i} \sim D_{-i}^+} \ps{ t_i \in R_{i,j}^{v_{-i}}}}}
\end{align}
where $\tilde{\varphi}_{ij}(t_{ij})$ represents Myerson's discrete ironed virtual value for the distribution $D_{ij}^+$.
\end{theorem}

\paragraph{Proof outline.}
Our approach will be to consider a discretization of the continuous type distribution $D$: $D^{\epsilon}$.  We define $D^{\epsilon}_{ij}$ to first sample $t_{ij} \sim D_{ij}$ and then output $t^{\epsilon}_{ij} = \epsilon^2 \cdot \lceil t_{ij}/\epsilon^2\rceil$. We see that $D^{\epsilon}_{ij}$ will have finite support $T^{\epsilon}$, as the support of $D_{ij}$ is bounded $\in [0,H]$.  Our approach is as follows.  Due to the coupling of samples from $D$ and $D^{\epsilon}$, we will be able to show that the revenue-optimal mechanism $OPT$ for values drawn from $D$ achieves approximately the same revenue as the revenue optimal mechanism $OPT^{\epsilon}$ for values drawn from $D^{\epsilon}$: $\Rev^D(OPT) \approx \Rev^{D^{\epsilon}}(OPT^{\epsilon})$.  Since $D^{\epsilon}$ has finite support, we will be able to apply \Cref{CaiThm} bounding the revenue of $OPT^{\epsilon}$ by its virtual welfare.  Then, one last argument on the coupled distributions will give that the virtual welfare upper bound for the discrete distribution is related to the desired virtual welfare bound for the continuous distribution.

\paragraph{Preference partitions are upwards-closed} Let $\{R_{i,j}\}_{i\in [n], j\in [m]}$ be a preference partition of the continuous type space $T$, which will also be a preference partition on the discrete subset $T^{\epsilon}$. Observe that preference partitions are always upwards-closed partitions, which is true due to the following argument:
Let $t_{ij},t_{ij}' \in T_{ij}$ with $t_{ij}' > t_{ij}$.  Say for some type vector $t_i=(t_{i1},\cdots,t_{ij},\cdots,t_{im})$ we have $t_i \in R_{i,j}$.  We want to show $t_i' = (t_{i1},\cdots,t_{ij}',\cdots,t_{im}) \in R_{i,j}$.  We see
\begin{align*}
    t_i \in R_{i,j} \Rightarrow &
    \pb{\begin{aligned}
    &\cU_{i,j}(t_{ij}) \geq \cU_{i,k}(t_{ik}) \quad \forall k \ne j \\ 
    &\cU_{i,j}(t_{ij}) > \cU_{i,k}(t_{ik}) \quad \forall k < j\\
    &\cU_{i,j}(t_{ij}) > 0 
    \end{aligned}}
    \Rightarrow 
    \pb{\begin{aligned}
    &\cU_{i,j}(t_{ij}') \geq \cU_{i,k}(t_{ik}) \quad \forall k \ne j\\
    &\cU_{i,j}(t_{ij}') > \cU_{i,k}(t_{ik}) \quad \forall k < j\\
    &\cU_{i,j}(t_{ij}') > 0
    \end{aligned}}
    \Rightarrow t_i' \in R_{i,j}
\end{align*}
since $\cU_{i,j}$ is non-decreasing, as desired. Thus we can apply \Cref{CaiThm} on the discretized type space and bound $\Rev^{D^{\epsilon}}(OPT^{\epsilon})$
\begin{align*}
    \Rev^{D^{\epsilon}}(OPT^{\epsilon}) \leq \underbrace{\sum_{i, j} \EE_{t_i \sim D_i^{\epsilon}}\ps{\pi^{OPT^{\epsilon}}_{ij}(t_i)\p{t_{ij}\cdot 1\pb{ t_i \not \in R_{i,j}}
+\tilde{\varphi}_{ij}^+(t_{ij})\cdot 1\pb{ t_i \in R_{i,j}}}}}_{\vw^{\epsilon}}
\end{align*}
where $\tilde{\varphi}_{ij}^+(t_{ij})$ is Myerson's discrete ironed virtual value for the distribution $D_{ij}^{\epsilon}$. In the latter we also used the fact that the preference partition of a player's type space is independent of the types of other players.

We conclude by separately relating the left-hand-side $\Rev^{D^{\epsilon}}(OPT^{\epsilon})$ to $\Rev^D(OPT)$ (\Cref{sec:discrete-lhs}), and the right-hand-side $\vw^{\epsilon}$ to its continuous counter-part $\vw$ (\Cref{sec:discrete-rhs}). In both cases, we show that the two quantities converge to each other as $\epsilon \rightarrow 0$, which implies the desired continuous upper bound.

\subsection{Relating $\Rev^D(OPT)$ to $\Rev^{D^{\epsilon}}(OPT^{\epsilon})$}\label{sec:discrete-lhs}

We make use of the following theorem.

\begin{theorem}\label{revDisc}[\cite{rubinstein2018simple}, \cite{DW12}]
Let $M^\downarrow$ be any BIC mechanism for additive bidders with values drawn from distribution $D^\downarrow$.  For all $i$, let $D_i^\downarrow$ and $D_i^\uparrow$ be any two distributions with coupled samples $t_i^\downarrow$ and $t_i^\uparrow$ such that $t_i^\uparrow \cdot x_i \geq t_i^\downarrow \cdot x_i$ for all feasible allocations $x \in F$. If $\delta_i = t_i^\uparrow - t_i^\downarrow$, then for any $\epsilon > 0$, there exists a BIC mechanism $M^\uparrow$ such that
$$\Rev^{D^\uparrow}(M^\uparrow) \geq (1 - \epsilon)\p{\Rev^{D^\downarrow}(M^\downarrow) - \frac{VAL(\delta)}{\epsilon}}$$
where $VAL(\delta)$ denotes the expected welfare of the VCG allocation when bidder $i$'s type is drawn according to the random variable $\delta_i$.
\end{theorem}

Using this theorem, we will be able to bound the gap between $\Rev^D(OPT)$ and $\Rev^{D^{\epsilon}}(OPT^{\epsilon})$.  We introduce $D^{-,\epsilon}_{ij}$, defined similarly to $D^{\epsilon}_{ij}$, that first samples $t_{ij} \sim D_{ij}$ and then outputs $t^{-,\epsilon}_{ij}= \epsilon^2 \cdot \p{\lceil t_{ij}/\epsilon^2\rceil-1}$.  Also, define $OPT^{-,\epsilon}$ to be the revenue-optimal mechanism for values drawn from $D^{-,\epsilon}$.  We will apply \Cref{revDisc} with $D$ as $D^\downarrow$ and $D^{\epsilon}$ as $D^\uparrow$ and $OPT$ as $M^\downarrow$ as well as with $D^{-,\epsilon}$ as $D^\downarrow$ and $D$ as $D^\uparrow$ and $OPT^{-,\epsilon}$ as $M^\downarrow$.  In both cases, due to the coupling, we will have the necessary $t_i^\uparrow \cdot x_i \geq t_i^\downarrow \cdot x_i$ for all $x$.  Additionally, we will have $\delta_{ij} \leq \epsilon^2$ for all $i,j$.  Thus, $VAL(\delta) \leq m\epsilon^2$ as the welfare contribution of any one item is at most $\epsilon^2$ for types $\delta_i$.  So, applying this theorem in these two settings gives, for some mechanisms $M^\uparrow_1$,$M^\uparrow_2$,
$$\Rev^{D^{\epsilon}}(OPT^{\epsilon}) \geq \Rev^{D^{\epsilon}}(M^\uparrow_1) \geq (1-\epsilon)(\Rev^D(OPT)-m\epsilon)$$
$$\Rev^{D}(OPT) \geq \Rev^{D}(M^\uparrow_2) \geq (1-\epsilon)(\Rev^{D^{-,\epsilon}}(OPT^{-,\epsilon})-m\epsilon)$$
Lastly, note that
$$\Rev^{D^{\epsilon}}(OPT^{\epsilon}) = \Rev^{D^{-,\epsilon}}(OPT^{-,\epsilon}) + m\epsilon$$
as every buyer values every item at exactly $\epsilon$ more in $D^{\epsilon}$ versus $D^{-,\epsilon}$.  For every BIC mechanism with values drawn from $D^{-,\epsilon}$, there is an analogous mechanism for values $D^{\epsilon}$ in which every bidders payment increases by exactly $\epsilon$ times the number of items they are expost allocated.  Thus, we have
$$\Rev^D(OPT) \in \ps{(1-\epsilon)(\Rev^{D^{\epsilon}}(OPT^{\epsilon})-2m\epsilon),\frac{\Rev^{D^{\epsilon}}(OPT^{\epsilon})}{1-\epsilon}+m\epsilon}$$
%\todo{so there is a typo in the Cai paper which gives $$\ps{(1-\epsilon)\Rev^{D^{-,\epsilon}}(OPT^{-,\epsilon})-m\epsilon,\frac{\Rev^{D^{\epsilon}}(OPT^{\epsilon})}{1-\epsilon}+\frac{m\epsilon}{1-\epsilon}}$$
%I went to the original RW paper to see what the real theorem stated and there was this extra $\eta$ term i didn't understand.  Anyways, we should double check this: https://arxiv.org/pdf/1501.07637.pdf Theorem 5.2 Page A20}
So, as $\epsilon \to 0$, we achieve discrete type distributions $D^{\epsilon}$ for which there exists mechanisms $OPT^{\epsilon}$ with revenue arbitrarily close to $\Rev^D(OPT)$.  

\subsection{Relating $\vw^{\epsilon}$ to $\vw$}\label{sec:discrete-rhs}
We can simulate a sample of the discrete distribution as follows: first sample $t_i \sim D_i$ from the continuous distribution and then let $t_i^+$ be the rounded discrete type in terms of $t_i$. That is, $t_{ij}^+ = \epsilon^2 \cdot \lceil t_{ij}/\epsilon^2\rceil$.  %\todo{Is this notation confusing? Is it necessary to write $t_i^+(t_i)$}  
We can then write the upper bound on $\Rev^{D^{\epsilon}}(OPT^{\epsilon})$ as:
\begin{align*}
    \vw^{\epsilon} = \sum_{i, j} \EE_{t_i \sim D_i}\ps{\pi^{OPT^{\epsilon}}_{ij}(t_i^+)\p{t_{ij}^+\cdot 1\pb{ t_i^+ \not \in R_{i,j}}
+\tilde{\varphi}_{ij}^+(t_{ij}^+)\cdot 1\pb{ t_i^+ \in R_{i,j}}}}
\end{align*}
Let $\pi_{ij}^M(t_i)=\pi^{OPT^{\epsilon}}_{ij}(t_i^+)$ and observe that $\pi^M$ is a feasible interim allocation as $OPT^{\epsilon}$ is a feasible mechanism and sampling $t_i$ from the continuous distribution and then rounding is identical to sampling from the discrete distribution.  We rewrite
\begin{align*}
    \vw^{\epsilon} =~& \sum_{i,j} \underbrace{\EE_{t_i \sim D_i}\ps{\pi^M_{ij}(t_i)\, t_{ij}^+\, 1\pb{ t_i^+ \not \in R_{i,j}}}}_{A^{\epsilon}} + \sum_{i, j} \underbrace{\EE_{t_i \sim D_i}\ps{\pi^M_{ij}(t_i)\, \tilde{\varphi}_{ij}^+(t_{ij}^+)\cdot 1\pb{ t_i^+ \in R_{i,j}}}}_{B^{\epsilon}}\\
\end{align*}
We denote with $A, B$ the corresponding continuous type terms where all plus signs are removed from the types. Moreover, in the $B$ term, the function $\tilde{\varphi}_{ij}$ is replaced by its non-negative version $\tilde{\varphi}_{ij}^*$.

{\setlength{\parindent}{0cm}\paragraph{Bounding $A^{\epsilon}$} We relate $A^{\epsilon}$ to $A$ as:}
\begin{align*}
    A^{\epsilon} - A =~& \EE_{t_i \sim D_i}\ps{\pi^M_{ij}(t_i)\, (t_{ij}^+-t_{ij})1\pb{ t_i^+ \not \in R_{i,j}}}
    +\EE_{t_i \sim D_i}\ps{\pi^M_{ij}(t_i)\, t_{ij}\, \p{1\pb{ t_i^+ \not \in R_{i,j}}-1\pb{ t_i \not \in R_{i,j}}}}
\end{align*}
We can bound $\EE_{t_i \sim D_i}\ps{\pi^M_{ij}(t_i)(t_{ij}^+-t_{ij})1\pb{ t_i^+ \not \in R_{i,j}}} \leq \epsilon^2$ as $\pi^M_{ij}(t_i) \leq 1$ and $t_{ij}^+-t_{ij} \leq \epsilon^2$.  Moreover:
\begin{align*}
\left| \EE_{t_i \sim D_i}\ps{\pi^M_{ij}(t_i)\, t_{ij}\, \p{1\pb{ t_i^+ \not \in R_{i,j}}-1\pb{ t_i \not \in R_{i,j}}}} \right| \leq H\cdot \sum_{k \ne j}\Pr_{t_i \sim D_i}\ps{ t_i \in R_{i,j} \wedge t_i^+ \in R_{i, k} }
\end{align*}
To upper bound this we prove the following lemma, whose proof we defer to \Cref{sec:prob-diff}.

\begin{lemma}\label{lem:measure-zero}
Let $D_i$ be an absolutely continuous distribution supported on a subset of $[0, H]$, with density upper bounded by $P$. Let $t_i^+$ denote the discrete type that corresponds to a rounded up version of each coordinate of $t_i$ to the closest multiple of $\epsilon^2$, i.e.: $t_{ij}^+ = \epsilon^2 \cdot \lceil t_{ij}/\epsilon^2\rceil$. If $\{R_{i,j}\}$ is a monotone preference partition of the continuous type space, then: 
\begin{equation}
\sum_{k \ne j}\Pr_{t_i \sim D_i}\ps{ t_i \in R_{i,j} \wedge t_i^+ \in R_{i, k}}\leq m^2\, H\, (2H+\epsilon^2)\, P\, \epsilon^2
\end{equation}
\end{lemma}
{\setlength{\parindent}{0cm}So, in total, we bound
$A^{\epsilon} \leq A + \epsilon^2\, (m^2\, H\, (2\, H + \epsilon^2)\, P+1) = A + o_{\epsilon}(1)$.
}

{\setlength{\parindent}{0cm}
\paragraph{Bounding $B^{\epsilon}$} Similarly, we decompose the term $B^{\epsilon}$:}
\begin{align*}
    B^{\epsilon} =~& \EE_{t_i \sim D_i}\ps{\pi^M_{ij}(t_i)\, (\tilde{\varphi}_{ij}^+(t_{ij}^+)-\tilde{\varphi}_{ij}(t_{ij}^+))\, 1\pb{ t_i^+ \in R_{i,j}}}\\
    & +\EE_{t_i \sim D_i}\ps{\pi^M_{ij}(t_i)\, \tilde{\varphi}_{ij}(t_{ij}^+)\p{1\pb{ t_i^+ \in R_{i,j}} - 1\pb{ t_i \in R_{i,j}}}}\\
    & +\EE_{t_i \sim D_i}\ps{\pi^M_{ij}(t_i)\, \tilde{\varphi}_{ij}(t_{ij}^+)\, 1\pb{ t_i \in R_{i,j}}}
\end{align*}
where $\tilde{\varphi}_{ij}$ represents Myerson's continuous ironed virtual value for the distribution $D_{ij}$.  In \cite{CDW16}, they prove that the discrete virtual value converges to the continuous virtual value for increasingly fine discretizations (Observation 9), $\lim_{\epsilon \to 0} \varphi_{ij}^+(v) = \varphi_{ij}(v)$ for all $v$.  We can easily extend this argument to ironed virtual values as $\varphi_{ij}^+$ converges to $\varphi_{ij}$ at all points. So,
$$\lim_{\epsilon \to 0} \tilde{\varphi}_{ij}^+(v) = \lim_{\epsilon \to 0} \max_{v \leq t} \varphi_{ij}^+(v) = \max_{v \leq t} \varphi_{ij}(v)= \tilde{\varphi}_{ij}(v)$$
Thus, 
\begin{equation*}
    \left|\EE_{t_i \sim D_i}\ps{\pi^M_{ij}(t_i)\, (\tilde{\varphi}_{ij}^+(t_{ij}^+)-\tilde{\varphi}_{ij}(t_{ij}^+))\, 1\pb{ t_i^+ \in R_{i,j}}}\right|=o_\epsilon(1)
\end{equation*}
and from \Cref{lem:measure-zero}:
\begin{equation*}
\left| \EE_{t_i \sim D_i}\ps{\pi^M_{ij}(t_i)\, \tilde{\varphi}_{ij}(t_{ij}^+)\p{1\pb{ t_i^+ \in R_{i,j}} - 1\pb{ t_i \in R_{i,j}}}} \right|\leq m^2\, H\, (2\, H + \epsilon^2)\, P\, \epsilon^2
\end{equation*}
as $\tilde{\varphi}_{ij}(t_{ij}^+) \leq t_{ij}^+$.  Thus, all that remains to show is that we can replace the $\tilde{\varphi}_{ij}(t_{ij}^+)$ term in \[\EE_{t_i \sim D_i}\ps{\pi^M_{ij}(t_i)\, \tilde{\varphi}_{ij}(t_{ij}^+)\, 1\pb{ t_i \in R_{i,j}}}\] with a $\tilde{\varphi}_{ij}(t_{ij})$.  Here, we make use of the relaxation of virtual value to positive virtual value: $\tilde{\varphi}_{ij}^*(t_{ij})=\max(\tilde{\varphi}_{ij}(t_{ij}),0)$.  Clearly, this upper bounds the virtual value.  It will give us a weaker result, but still a meaningful bound.  We have
\begin{align*}
    \EE_{t_i \sim D_i}\ps{\pi^M_{ij}(t_i)\, \tilde{\varphi}_{ij}(t_{ij}^+)\, 1\pb{ t_i \in R_{i,j}}}
    \leq~& \EE_{t_i \sim D_i}\ps{\pi^M_{ij}(t_i)\, \tilde{\varphi}^*_{ij}(t_{ij}^+)\, 1\pb{ t_i \in R_{i,j}}}\\
    =~& \EE_{t_i \sim D_i}\ps{\pi^M_{ij}(t_i)\, (\tilde{\varphi}^*_{ij}(t_{ij}^+)-\tilde{\varphi}^*_{ij}(t_{ij}))\, 1\pb{ t_i \in R_{i,j}}} + B\\
    \leq~& \EE_{t_i \sim D_i}\ps{\tilde{\varphi}^*_{ij}(t_{ij}^+)-\tilde{\varphi}^*_{ij}(t_{ij})} + B\\
    \leq~& \EE_{t_i \sim D_i}\ps{\tilde{\varphi}^*_{ij}(t_{ij}^+)-\tilde{\varphi}^*_{ij}(t_{ij}^-)} + B
\end{align*}
where $t_{ij}^- = \epsilon^2 \cdot \p{\lceil t_{ij}/\epsilon^2\rceil-1}$.  This is true since the \emph{ironed} virtual value function is non-decreasing.  We can view the discretization as a breaking up of $T_{ij}$ into segments of length $\epsilon^2$ and $\tilde{\varphi}^*_{ij}(t_{ij}^+)-\tilde{\varphi}^*_{ij}(t_{ij}^-)$ will be the difference in the endpoints of the interval containing $t_{ij}$. Moreover, making use of the fact that $\tilde{\varphi}^*_{ij}$ is a non-decreasing function with range $\subseteq [0,H]$ as $\tilde{\varphi}_{ij}(t_{ij}) \leq t_{ij}$:
\begin{align*}
    \EE_{t_i \sim D_i}\ps{\tilde{\varphi}^*_{ij}(t_{ij}^+)-\tilde{\varphi}^*_{ij}(t_{ij}^-)}
    \leq \epsilon + H \cdot \Pr_{t_i \sim D_i}\ps{\tilde{\varphi}^*_{ij}(t_{ij}^+)-\tilde{\varphi}^*_{ij}(t_{ij}^-) > \epsilon}
\end{align*}
Moreover, due to the monotonicity of $\tilde{\varphi}^*_{ij}$, there can only be at most $H/\epsilon$ segments of $T_{ij}$ with endpoints differing by at least $\epsilon$.  Again making use of the fact that $D_{ij}$ is atomless and there is some finite upper bound $P$ on its density function, we can argue that the probability of $t_{ij} \sim D_{ij}$ belonging to any specific interval is $\leq P\epsilon^2$.  Thus,
$$\Pr_{t_i \sim D_i}\ps{\tilde{\varphi}^*_{ij}(t_{ij}^+)-\tilde{\varphi}^*_{ij}(t_{ij}^-) > \epsilon} \leq (H/\epsilon)\cdot P\epsilon^2 = HP\epsilon$$
and so, $\EE_{t_i \sim D_i}\ps{\tilde{\varphi}^*_{ij}(t_{ij}^+)-\tilde{\varphi}^*_{ij}(t_{ij}^-)} \leq (H^2P+1)\, \epsilon$.
Putting all this together, we have
$B^{\epsilon} \leq B + o_\epsilon(1)$.

{\setlength{\parindent}{0cm}
\paragraph{Concluding} Combining the facts that $A^{\epsilon}\leq A + o_{\epsilon}(1)$ and $B^{\epsilon} \leq B + o_\epsilon(1)$, yields:
\begin{equation*}
    VW^{\epsilon}
    \leq \sum_i \EE_{t_i \sim D_i}\ps{\sum_j \pi_{ij}^M(t_i)\p{t_{ij}\cdot 1\pb{t_i \not \in R_{i, j}}
+\tilde{\varphi}_{ij}^*(t_{ij})\cdot 1\pb{ t_i \in R_{i, j}}}} + o_\epsilon(1)
\end{equation*}
giving the desired upper bound as $\epsilon \to 0$.
}

\subsection{Proof of \Cref{lem:measure-zero}}\label{sec:prob-diff}
\begin{proof}
In order to have $t_i \in R_{i,j}$ and $t_i^+ \in R_{i, k}$, we must have $\cU_{i,j}(t_{ij}) \geq \cU_{i,k}(t_{ik})$ and $\cU_{i,j}(t_{ij}^+) < \cU_{i,k}(t_{ik}^+)$ in the event $j<k$.  Similarly, we must have $\cU_{i,j}(t_{ij}) > \cU_{i,k}(t_{ik})$ and $\cU_{i,j}(t_{ij}^+) \leq \cU_{i,k}(t_{ik}^+)$ in the event $j>k$.  We assume,without loss of generality, that $j<k$ as the argumentation is symmetric in both cases.

We think about the two-dimensional plane $T_{ij} \times T_{ik}$ of possible values $(t_{ij},t_{ik})$.  We can view the discretization $(t_{ij}^+,t_{ik}^+)$ as a division of this plane into a grid of squares of side length $\epsilon^2$.  Here, $(t_{ij}^+,t_{ik}^+)$ represents the upper right corner of whichever square $(t_{ij},t_{ik})$ belongs to.  We also consider a partitioning of this plane into the set of points for which $\cU_{i,j}(t_{ij}) \geq \cU_{i,k}(t_{ik})$ and the set of points for which $\cU_{i,j}(t_{ij}) < \cU_{i,k}(t_{ik})$.  In order to have $\cU_{i,j}(t_{ij}) \geq \cU_{i,k}(t_{ik})$ and $\cU_{i,j}(t_{ij}^+) < \cU_{i,k}(t_{ik}^+)$, we must have the border of this partition pass through the square containing $(t_{ij},t_{ik})$.  However, we show that only a small number of squares will contain a piece of this border, enabling us to bound the probability of such an event as $\epsilon \rightarrow 0$.

\paragraph{Proof intuition}
The border of any monotone preference partition, when projected on the two dimensional plane $T_{ij}\times T_{ik}$ of the types $(t_{ij}, t_{ik})$ for two items, must be a curve that corresponds to a monotone non-decreasing function of $t_{ij}$. Thus any two squares that are in the $x+y=u$ diagonal (for some $u$), cannot contain points from both partitions as that would imply that there is a point of the border in both squares, which would subsequently imply that these two points violate the monotonicity of the border. Since there are at most $O(H/\epsilon^2)$ diagonals, there can be at most $O(H/\epsilon^2)$ squares that can be problematic, each with density at most $P\, \epsilon^4$. In total a probability mass of types of at most $O(H\, \epsilon^2)\rightarrow 0$, can be problematic (see \Cref{fig:low-prob}).

\begin{figure}[htpb]
  \centering
  \pgfdeclarelayer{level0}
    \pgfdeclarelayer{level1}
    \pgfdeclarelayer{level2}
    \pgfsetlayers{main,level0,level1,level2}
    \begin{tikzpicture}[scale=.7]
      \begin{pgfonlayer}{level0} % Level 0
        \draw[help lines] (0,0) grid[step=.5] (5,5); % Base grid
        \draw[very thick, scale=5] (0,0) grid (1,1); % Darker lines to mark e.g. box boundaries
      \end{pgfonlayer}

      \begin{pgfonlayer}{level1} % Level 2
      
      \draw[fill=black!10] (4.5,4.5) rectangle (5,5);
      \draw[fill=black!10] (4,4.5) rectangle (4.5,5);
      \draw[fill=black!10] (3.5,4.5) rectangle (4,5);
      \draw[fill=black!10] (3.5,4) rectangle (4,4.5);
      \draw[fill=black!10] (3,4) rectangle (3.5,4.5);
      \draw[fill=black!10] (3,3.5) rectangle (3.5,4);
      \draw[fill=black!10] (3,3) rectangle (3.5,3.5);
      \draw[fill=black!10] (2.5,3) rectangle (3,3.5);
      \draw[fill=black!10] (2.5,2.5) rectangle (3,3);
      \draw[fill=black!10] (2.5,2) rectangle (3,2.5);
      \draw[fill=black!30] (2,2) rectangle (2.5,2.5);
      \draw[fill=black!10] (2,1.5) rectangle (2.5,2);
      \draw[fill=black!10] (2,1) rectangle (2.5,1.5);
      \draw[fill=black!10] (1.5,1) rectangle (2,1.5);
      \draw[fill=black!10] (1.5,.5) rectangle (2,1);
      \draw[fill=black!10] (1,.5) rectangle (1.5,1);
      \draw[fill=black!10] (1,0) rectangle (1.5,.5);
      \draw[fill=black!10] (.5,0) rectangle (1,.5);
      \draw[fill=black!10] (0,0) rectangle (.5,.5);

      \draw[fill=black!10] (7, 4.5) rectangle (7.5,5);
      \node[] at (8.5, 4.7) {\footnotesize{bad cells}};
      \draw[pattern=north west lines] (7, 3.5) rectangle (7.5,4);
      \node[] at (9.3, 3.7) {\footnotesize{example diagonal}};
      \draw[fill=black] (7, 2.8) rectangle (7.5,2.85);
      \node[] at (8.4, 2.9) {\footnotesize{border}};
      \end{pgfonlayer}
    
      \begin{pgfonlayer}{level1} % Level 2
      
      \draw[pattern=north west lines] (4,0) rectangle (4.5,.5);
      \draw[pattern=north west lines] (3.5,0.5) rectangle (4,1);
      \draw[pattern=north west lines] (3,1) rectangle (3.5,1.5);
      \draw[pattern=north west lines] (2.5,1.5) rectangle (3,2);
      \draw[pattern=north west lines] (2,2) rectangle (2.5,2.5);
      \draw[pattern=north west lines] (1.5,2.5) rectangle (2,3);
      \draw[pattern=north west lines] (1,3) rectangle (1.5,3.5);
      \draw[pattern=north west lines] (.5,3.5) rectangle (1,4);
      \draw[pattern=north west lines] (0,4) rectangle (.5,4.5); 
      
      \draw[thick] (0,0) to[out=0,in=190] (5,5);
      \node[] at (5.3, -0.1) {\footnotesize{$t_{ij}$}};
      \node[] at (-0.5, 4.9) {\footnotesize{$t_{ik}$}};
      \end{pgfonlayer}
    \end{tikzpicture}
    \caption{Pictorial representation of proof arguments.}\label{fig:low-prob}
\end{figure}
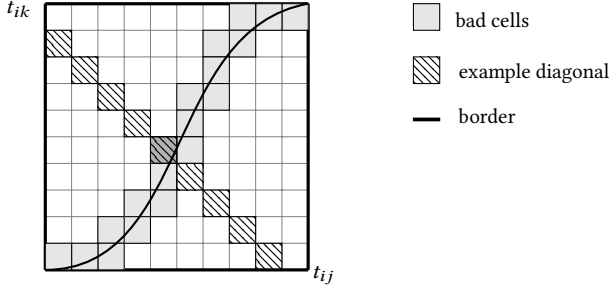

\paragraph{Formal argument} We can index the grid of squares as an ordered pair $(x,y)$ where $x$ and $y$ are integers in the range $[1,\lceil H/\epsilon^2 \rceil]$.  Square $(x,y)$ contains the points $((x-1)\epsilon^2,x\epsilon^2]\times ((y-1)\epsilon^2,y\epsilon^2]$.  In the edge cases, index 1 corresponds to $[0,\epsilon^2]$ inclusive and index $\lceil H/\epsilon^2 \rceil$ corresponds to $[\epsilon(\lceil H/\epsilon^2 \rceil-1),H]$.

We claim that, for any two squares $(x_1,y_1),(x_2,y_2)$ containing points from both sides of the partition, we must have $x_1+y_1 \ne x_2+y_2$.  Assume for the sake of contradiction that $x_1+y_1 = x_2+y_2$ and WLOG $x_1<x_2$, $y_1>y_2$.  Say we had $(t_{ij,1},t_{ik,1})$ in square $(x_1,y_1)$ with $\cU_{i,j}(t_{ij,1}) \geq \cU_{i,k}(t_{ik,1})$ and $(t_{ij,2},t_{ik,2})$ in square $(x_2,y_2)$ with $\cU_{i,j}(t_{ij,2}) < \cU_{i,k}(t_{ik,2})$.  We cannot have $\cU_{i,k}(t_{ik,1})<\cU_{i,k}(t_{ik,2})$.  We must have $t_{ik,1}>t_{ik,2}$ since $y_1>y_2$ and $\cU_{i,k}$ is non-decreasing.  However, we cannot have $\cU_{i,k}(t_{ik,1}) \geq \cU_{i,k}(t_{ik,2})$ as that would imply $\cU_{i,j}(t_{ij,1})>\cU_{i,j}(t_{ij,2})$.  We must have $t_{ij,1}<t_{ij,2}$ since $x_1<x_2$ and $\cU_{i,j}$ is non-decreasing, a contradiction.

Thus, the partition border can only pass through one square along the diagonal of squares $(x,y)$ satisfying $x+y=u$.  Since $x,y$ are integers in the range $[1,\lceil H/\epsilon^2 \rceil]$, we have $x+y \in [2, 2\lceil H/\epsilon^2 \rceil]$.  Therefore, the partition border passes through at most $2H/\epsilon^2 + 1$ squares.

Then, since $D_i$ is a bounded distribution and is absolutely continuous with respect to the Lebesgue measure, the probability density function of $D_{ij}$ is bounded for every $j$.  Thus, there is some finite $P$ that upper bounds the PDF of the joint distribution $D_{ij}\times D_{ik}$ for every pair $j,k$.  So, the probability of $(t_{ij},t_{ik})$ belonging to any specific square is at most $P(\epsilon^2)^2$.  Therefore, the probability that $(t_{ij},t_{ik})$ belongs to a square containing a piece of the partition border is $\leq 2HP\epsilon^2+P\epsilon^4$.  So,

\begin{align*}
   \sum_{k \ne j}\Pr_{t_i \sim D_i}\ps{ t_i \in R_{i,j} \wedge t_i^+ \in R_{i, k} }&\leq \sum_{k \ne j} (2HP\epsilon^2+ P\epsilon^4)
    \leq m^2\, H\, (2\, H + \epsilon^2)\, P\, \epsilon^2
\end{align*}

as desired.
\end{proof}

\section{Proof of Lemma~\ref{lem:rand-ea-eq}}\label{app:rand-ea-eq}

\begin{blanklemma2}
 Let $\bb$, be any mixed BNE of the $\mathsf{rand}-EA(e)$ auction, with entry fees $e=\{e_i\}$, for type vector distribution $D=\times_{j\in[m]} D_j$, where $D_j=\times_i D_{ij}$. Let $\tilde{\bb}_i^j: T_{ij} \to \Delta(A_{ij})$, denote the marginal action distribution of player $i$ on item $j$ conditional only on her type $t_{ij}$ for item $j$. Then, $\tilde{\bb}^j = \{\tilde{\bb}_i^j\}$ is a mixed BNE of the $A$ auction for item $j$, when run in isolation. Moreover:
\begin{equation}
    \efrev^{\bb}(\REA(e)) \geq (1-\delta)\efrev^{\tilde{\bb}}(A,D,e)
\end{equation}
where $\efrev^{\bb}(\REA(e))$ represents the revenue of the $\REA$ auction solely coming from the entry fees and
$$\efrev^{\tilde{\bb}}(A,D,e) = \sum_i e_i\Pr_{t_i \sim D_{i}}\ps{\sum_j u^{\tilde{\bb}}_{ij}(t_{ij}) \geq e_i}$$
\end{blanklemma2}

\begin{proof}[Proof of \Cref{lem:rand-ea-eq}]
Consider a mixed equilibrium strategy $\bb=\{\bb_i\}_{i\in [n]}$ of the $\mathsf{rand}-EA$ auction. Let $\bb_i^j: T_i \to \Delta(A_{ij})$, \footnote{Where $\Delta(A_{ij})$ denotes the set of distributions over bids submitted by player $i$ to the auction for item $j$.} denote the mapping from a type $t_i\in T_i$ to the marginal distribution of actions of bidder $i$ for item $j$, conditional on type $t_i$, under the mixed equilibrium $\bb$. More concretely, the (probability) density function of distribution $\bb_i^j(t_i)$ is given by
\[
p_{\bb_i^j (t_{i})}(b) = \mathbb{E}_{(z_i, a_i)\sim \bb_i(t_i)}\ps{1\pb{a_i^j = b}}
\]
Moreover, let $\tilde{\bb}_i^j: T_{ij} \rightarrow \Delta(A_{ij})$ denote the marginal distribution of actions on the auction for item $j$ conditional only on her type $t_{ij}$ for item $j$ and marginalizing her types for other items. More concretely, the (probability) density function of distribution $\tilde{\bb}_i^j(t_{ij})$ is given by
\[
p_{\tilde{\bb}_i^j (t_{ij})}(b) = \EE_{t_{i, -j} \sim D_{i, -j}}\ps{\EE_{(z_i, a_i)\sim \bb_i(t_i)}\ps{1\pb{a_i^j = b}}}
\]

The interim utility of bidder $i$ with type $t_i$, in the $\mathsf{rand}-EA$ auction, is given by
\[
u_{i}^\bb(t_i) = \sum_{j=1}^m \EE_{t_{-i}\sim D_{-i}}\ps{\EE_{(z, a)\sim \bb(t)}\ps{(z_i\,(1-\delta) + \delta) \cdot u^*_{ij} (t_i;a)}} - \EE_{(z_i, a_i)\sim \bb_i(t_i)}\ps{z_i\, (1-\delta) \, e_i},
\]where $u^*_{ij} (t_i;a)$ denotes the ex-post utility of bidder $i$ in the auction for item $j$ under bid profile $a$. 
Since the ex-post utility $u_{ij}^*$ depends only on $t_{ij}$ and the bid profile $a^j$ for item $j$, we can re-write the above expression for interim utility as
\begin{align*}
u_{i}^\bb(t_i) =~& \sum_{j=1}^m \EE_{t_{-i}\sim D_{-i}}\ps{\EE_{(z, a)\sim \bb(t)}\ps{(z_i\,(1-\delta) + \delta) \cdot u^*_{ij} (t_{ij};a^j)}} - \EE_{(z_i, a_i)\sim \bb_i(t_i)}\ps{z_i\, (1-\delta) \, e_i}
\end{align*}
For simplicity let $G_{-i}^j\in \Delta(A_{-i, j})$ denote the distribution of other player actions at the auction of item $j$ under the mixed BNE $\bb$ of the $\mathsf{rand}-EA$ auction. Moreover, observe that this is the same distribution as first drawing a random type $t_{i',j}$ of each opponent $i'$ for item $j$ and then drawing an action for that player from the marginal distribution $\tilde{b}_{i'}^j(t_{i', j})$. Then: 
\begin{align*}
u_{i}^\bb(t_i) =~& \sum_{j=1}^m \EE_{(z_i, a_i)\sim \bb_i(t_i)}\ps{(z_i\,(1-\delta) + \delta) \cdot \EE_{a_{-i}^j \sim G_{-i}^j}\ps{u^*_{ij} (t_{ij};a^j)}} - \EE_{(z_i, a_i)\sim \bb_i(t_i)}\ps{z_i\, (1-\delta) \, e_i}
\end{align*}

Let $U_{ij}(t_{ij}; a_i^j) = \EE_{a_{-i}^j \sim G_{-i}^j}\ps{u^*_{ij} (t_{ij};a^j)}$, then:
\begin{align*}
u_{i}^\bb(t_i) =~& \sum_{j=1}^m \underbrace{\EE_{(z_i, a_i)\sim \bb_i(t_i)}\ps{(z_i\,(1-\delta) + \delta) \cdot U_{ij}(t_{ij}; a_i^j)}}_{A_{ij}(t_i)} - \EE_{(z_i, a_i)\sim \bb_i(t_i)}\ps{z_i\, (1-\delta) \, e_i}
\end{align*}

Let $u_{ij}^{\bb}(t_i) = \max_{a_i^j\in A_{ij}} U_{ij}(t_{ij}; a_i^j)$. Now suppose that the distribution $b_i(t_i)$ submits with probability $\rho>0$ actions $a_i^j$ that achieve utility $U_{ij}(t_{ij}; a_i^j) \leq u_{ij}^{\bb}(t_i) - \epsilon$ for $\epsilon > 0$. Then observe that the player can deviate and strictly increase their utility by submitting action $\arg\max_{a_i^j\in A_{ij}} U_{ij}(t_{ij}; a_i^j)$, whenever they would have submitted any such sub-optimal action $\tilde{a}_i^j$. This is a strictly improving deviation since, it leads to an improvement of at least $\delta\, \epsilon\, \rho$. Thus we have that, when $a_i$ is drawn from distribution $b_i(t_i)$, then with probability $1$: $U_{ij}(t_{ij}; a_i^j)=u_{ij}^{\bb}(t_i)$. We can then re-write $A_{ij}(t_i)$:
\begin{align*}
    A_{ij}(t_i) =~& \EE_{(z_i, a_i)\sim \bb_i(t_i)}\ps{(z_i\,(1-\delta) + \delta) \cdot u_{ij}^{\bb}(t_i)}\\
    =~& \p{\EE_{(z_i, a_i)\sim \bb_i(t_i)}\ps{z_i}\,(1-\delta) + \delta} \cdot u_{ij}^{\bb}(t_i)\\
    =~& \p{\EE_{(z_i, a_i)\sim \bb_i(t_i)}\ps{z_i}\,(1-\delta) + \delta} \cdot \EE_{(z_i, a_i)\sim \bb_i(t_i)}\ps{U_{ij}(t_{ij}; a_i^j)}
\end{align*}
Now observe, that since $U_{ij}(t_{ij}; a_i^j)$ is independent of $t_{i, -j}$, we then have that:
\begin{equation}
    \EE_{(z_i, a_i)\sim \bb_i(t_i)}\ps{U_{ij}(t_{ij}; a_i^j)} = \EE_{a_i^j \sim \tilde{\bb}_i^j(t_{ij})}\ps{U_{ij}(t_{ij}; a_i^j)} 
\end{equation}
Now observe, that the latter is the interim utility of player $i$ in a single item auction for item $j$, where players use bid strategies $\tilde{\bb}^j=\pb{\tilde{\bb}_i^j}_{i\in [n]}$, denoted as $u_{ij}^{\tilde{\bb}^j}(t_{ij})$. Thus we can write a player's interim utility in the $\mathsf{rand}-EA$ auction in terms of the latter interim utility as:
\begin{align*}
u_{i}^\bb(t_i) =~& \sum_{j=1}^m \p{\EE_{(z_i, a_i)\sim \bb_i(t_i)}\ps{z_i}\,(1-\delta) + \delta} \cdot u_{ij}^{\tilde{\bb}^j}(t_{ij}) - \EE_{(z_i, a_i)\sim \bb_i(t_i)}\ps{z_i\, (1-\delta) \, e_i}
\end{align*}
If we denote with $q_i^{\bb_i}(t_i) = \EE_{(z_i, a_i)\sim \bb_i(t_i)}\ps{z_i}$, the marginal probability of entry with type $t_i$ under the mixed BNE $\bb$, then:
\begin{align*}
u_{i}^\bb(t_i) =~& \sum_{j=1}^m \p{q_i^{\bb_i}(t_i)\,(1-\delta) + \delta} \, u_{ij}^{\tilde{\bb}^j}(t_{ij}) - q_i^{\bb_i}(t_i)\, (1-\delta) \, e_i \\
=~& q_i^{\bb_i}(t_i)\, (1-\delta)\, \p{\sum_{j=1}^m u_{ij}^{\tilde{\bb}^j}(t_{ij}) - e_i} + \delta\, \sum_{j=1}^m u_{ij}^{\tilde{\bb}^j}(t_{ij})
\end{align*}

Now we argue that the marginal distribution mappings $\tilde{\bb}^j=\pb{\tilde{\bb}_i^j}_{i\in [n]}$, must constitute a mixed BNE of the single item auction $A$ for item $j$, if run in isolation. Suppose that this was not the case. This means that there is some player $i$ that has a profitable deviating strategy, i.e. that has some action $\tilde{a}_i^j$, such that for some $\epsilon > 0$:
\begin{equation}
    u_{ij}^{\tilde{\bb}^j}(t_{ij}) \leq U_{ij}(t_{ij}; \tilde{a}_i^j) - \epsilon 
\end{equation}
However, in that case there is a profitable deviation of player $i$ in the $\mathsf{rand}-EA$ auction, since if player $i$ was always submitting action $\tilde{a}_i^j$ on item $j$, instead of her prior bid, she could increase her interim utility by at least $\delta\, \epsilon$.

Finally, observe that a player enters the $\mathsf{rand}-EA$ auction deterministically whenever:
\begin{align*}
    \sum_{j=1}^m u_{ij}^{\tilde{\bb}^j}(t_{ij}) - e_i > 0
\end{align*}
otherwise there is a profitable deviation. Thus the probability of entry is at least:
\begin{align*}
    \Pr_{t_i\sim D_i}\left[\sum_{j=1}^m u_{ij}^{\tilde{\bb}^j}(t_{ij}) - e_i > 0\right]
\end{align*}
%Observe that this is equal to the entry probability in the EA auction with $\pb{\tilde{\bb}^j}$-simulating ghost bidders and entry fees $\{e_i\}_{i\in [n]}$ at the focal equilibrium $\tilde{\bb}$. 

Thus the entry fee revenue collected by the $\mathsf{rand}-EA$ auction at \emph{any mixed BNE} equilibrium $\bb$ is at least:
\begin{align}
    \efrev^{\bb}(\mathsf{rand}-\EA(e)) &\geq (1-\delta)\sum_i e_i\Pr_{t_i\sim D_i}\left[\sum_{j=1}^m u_{ij}^{\tilde{\bb}^j}(t_{ij}) - e_i > 0\right]\\
    &=(1-\delta)\cdot \efrev^{\tilde{\bb}}(A,D,e)
\end{align}

%\efrev^b(A,D,e) = \sum_i e_i\Pr_{t_i \sim D_{i}}\ps{\sum_j u^b_{ij}(t_{ij}) \geq e_i}

\end{proof}

\section*{Acknowledgements}
The first author was supported by NSF Awards IIS-1741137, CCF-1617730 and CCF-1901292, by a Simons Investigator Award, by the DOE PhILMs project (No. DE-AC05-76RL01830), and by the DARPA award HR00111990021. The fifth author was supported in part by a Simons Investigator Award and NSF Award CCF 1715187. This work was done in part while the first author was visiting Microsoft Research-New England.

\bibliographystyle{alpha}
\bibliography{ec2020}

\end{document}